\documentclass[12pt]{article}

\usepackage{hyperref}

\usepackage{amsmath} 
\usepackage{amssymb}
\usepackage{amsfonts}
\usepackage{amsthm}
\usepackage{graphicx}

\newtheorem{theorem}{Theorem}
\newtheorem{lemma}{Lemma}
\newtheorem{corollary}{Corollary}

\theoremstyle{definition}

\begin{document}
	\begin{center}
		\Large
	\textbf{Non-Markovian evolution of multi-level system interacting with several reservoirs. Exact and approximate}\footnote{The research was supported by RSF (project No. 17-71-20154).}	
	
		\large 
		\textbf{A.E. Teretenkov}\footnote{Steklov Mathematical Institute of Russian Academy of Sciences,
			ul. Gubkina 8, Moscow 119991, Russia; 
			Lomonosov Moscow State University, Leninskie Gory 1, Moscow 119991, Russia\\ E-mail:\href{mailto:taemsu@mail.ru}{taemsu@mail.ru}}
		\end{center}
		
			\footnotesize
			An exactly solvable model for the multi-level system interacting with several reservoirs at zero temperatures is presented. Population decay rates and decoherence rates predicted by exact solution and several approximate master equations, which are widespread in physical literature, are compared. The space of parameters is classified with respect to different inequities between the exact and approximate rates.
			\normalsize

	\section{Introduction}
	
	In this work we consider the evolution of the multi-level system each level of which interacts with its own bosonic reservoir at zero temperature. For the simplicity we assume that the reservoirs are similar and the coupling of each level to its reservoir is also similar (the explicit mathematical model is described  in  Section \ref{sec:statement}). The main aim of the work is the comparison of the exact evolution of the reduced density matrix of the system (obtained by the pseudomode method, see Section \ref{sec:pseudomode})  with the approximate evolution defined by master equations which are widespread in the physical literature. Namely, we consider the Nakajima-Zwanzig equation in the Born approximation, the non-Markovian Redfield equation and the Markovian Redfield equation. Usually in physical literature equations for the reduced density matrix are derived in the following way. The initial Liouville-von~Neumann (linear differential) equation is reduced to the Nakajima-Zwanzig (linear integro-differential) equation \cite{Nakajima58, Zwanzig60} (or some equivalent equation obtained by exclusion of reservoirs degrees of freedom). The Nakajima--Zwan\-zig is an exact equation for the reduced density matrix. Then the following four assumptions are subsequently done:
	\begin{enumerate}
		\item Born approximation \cite[p. 131]{Breuer02}, \cite[p. 7]{Carmichael13}. This equation is also integro-differential, which allows one to consider it as a non-Markovian one \cite{Chruscinski10,ChruscinskiKossakowski10}. Sometimes this approximation is also called the Redfield approximation or the second-order approximation  \cite[p. 249]{Valkunas13}, \cite[Subsec. 11.2]{Amerongen18}, \cite{Singh12}. 
		\item Assumption that the reduced density matrix inside the integral  could be taken at the same time as outside, which leads to the non-Markovian Redfield equation \cite{Redfield65}, \cite[p. 132]{Breuer02}, \cite{Ishizaki08}. Actually, this approximation is very close to Markovianity but only on the long times. 
		\item The full Markovian approximation, which leads to the Markovian Redfield equation \cite[p. 141]{May08}, \cite[p. 132]{Breuer02} , \cite[p. 7]{Carmichael13}, . 
		\item  Secular approximation \cite[p. 132]{Breuer02}, \cite[p. 145]{May08}, which leads to Markovian equations with the Gorini--Kossakowski--Sudarshan--Lindblad (GKSL) generator \cite{Gorini76, Lindblad76}.
	\end{enumerate}
	
	At the same time only the equation which obtained after all four assumptions has mathematically strict justification \cite{Davies1974, Accardi2002}. And its derivation goes back to  \cite{Krylov70, VanHove55} and is based on the van Hove--Bogolyubov scaling \cite[Sec. 1.8]{Accardi2002}. Moreover, only the last equation guarantees positivity and even complete positivity \cite{Gorini76, Lindblad76}. At the same time the Redfield equation can violate the positivity, which indeed could be fixed by applying the slippage-operators to the initial conditions \cite{Suarez92, Pierre99, Anderloni07} or by some other methods \cite{Giovannetti19}. On the other hand, in physical calculations the Redfield equation without secular approximation is frequently used \cite{Novoderezhkin03, Novoderezhkin04}. Some advantage of the Redfield equation in comparison with GKSL equations of the secular approximation are discussed in \cite{Purkayastha16, Dodin18}. For the damped oscillator the Redfield equation can both be translation invariant and have canonical equilibrium state \cite{Kohen97} rather than GKSL equation \cite{Lindblad76Br}.  If the aim is to take into account the non-Markovian effects then only Born approximation is usually done \cite{Singh12}. Thus, it is suggested in physical literature that the equations of the first one or two approximations could be more accurate than the GKSL equation which needs all four approximations. So it is natural to study the exactly solvable model to compare its prediction with the approximate ones.
	
	In Section \ref{sec:statement} we present the initial problem for the system and the reservoirs. It is important to mention that we introduce the Hamiltonian which is already in the rotating wave approximation form. Hence we leave the discussion of this approximation by itself out of the range of our study. At the same time, if one derives GKSL equations by the stochastic limit approach, then it is not necessary to assume this approximation and its a corollary of the van Hove--Bogolyubov scaling \cite{Accardi2002}.
	
	In Section \ref{sec:pseudomode} we present the pseudomode approach and obtain the exact evolution of the reduced density matrix. This approach was developed in \cite{Imamoglu94, Garraway96,Garraway97,Garraway97a,Dalton01,Garraway06}. In \cite{Teret19} we have shown that the Friedrichs model \cite{Friedrichs48}  naturally arises as an intermediate step in this approach. In a quite general form the description of the non-Markovian evolution in the Friedrichs approximation was discussed in \cite{Kossakowski07}.
	
	In Sections \ref{sec:NakZwanBorn} and \ref{sec:Redfield} we present the Nakajima-Zwanzig equation in Born approximation and the Redfield equation (both non-Markovian and Markovian) for our special case, accordingly, and solve them. In our case the Markovian Redfield equation appears to be identical to the one with the secular approximation. 
	
	Finally, in Section \ref{sec:Compare} we compare the solutions for exact and approximate equations. In Conclusions we summarize our results and suggest the directions for the further studies.
	
	\section{Schroedinger equation for system and reservoirs}\label{sec:statement}
	We consider the evolution in the Hilbert space
	\begin{equation*}
	\mathcal{H} \equiv (\mathbb{C}\oplus\mathbb{C}^{N}) \otimes \bigotimes\limits_{i=1}^N \mathfrak{F}_b(\mathcal{L}^2(\mathbb{R})).
	\end{equation*}
	Here $  \mathbb{C}\oplus\mathbb{C}^{N}$ is a $ (N+1) $-dimensional Hilbert space with a pointed one-dimensional subspace which corresponds to the degrees of freedom of the $ (N+1) $-level system. Let $ | i \rangle , i = 0,  1, \ldots, \ N $ be an orthonormal basis in such a space and $ | 0 \rangle $ correspond (be collinear) to the pointed subspace.  $ \mathfrak{F}_b(\mathcal{L}^2(\mathbb{R})) $ are bosonic Fock spaces which describe the reservoirs. There is one reservoir for each excited level of the system. Let $ | \Omega \rangle $ be a vacuum vector for the reservoirs. Let us also introduce the creation and annihilation operators which satisfy the canonical commutation relations: $ [b_{k,i}, b_{k',j}^{\dagger}] = \delta_{ij} \delta (k - k')$, $[b_{k,i}, b_{k',j}] = 0 $, $ b_{k,i} | \Omega \rangle = 0$.
	
	We consider the system Hamiltonian in the general form
	\begin{equation}\label{eq:H_S}
	\hat{H}_S  = \sum_{i} \varepsilon_i |i \rangle \langle i| + \sum_{i\neq j} J_{ij} |i \rangle \langle j| =  0 \oplus H_S, \qquad i,j =1 ,\ldots,N,
	\end{equation}
	without assumption that $ \hat{H}_S $ is diagonalized in the basis $ |i \rangle  $. From the physical point of view  $ |i \rangle  $ plays the role of local basis \cite{Levy14, Trushechkin16}. The only restriction is that it is non-zero only in the excited subspace i.e. in the subspace which is an orthogonal complement to the pointed one-dimensional subspace. The restriction of $ \hat{H}_S  $ to the excited subspace is denoted by $ H_S $ in formula \eqref{eq:H_S}.
	
	The reservoir Hamiltonian is a sum of identical Hamiltonians of free bosonic fields (with the identical dispersion relation $ \omega_k $)
	\begin{equation}\label{eq:H_B}
	\hat{H}_B =\sum_{i=1}^N \int \omega_k  b_{k,i}^{\dagger} b_{k,i} d k .
	\end{equation}
	
	The interaction is described by the following Hamiltonian
	\begin{equation}\label{eq:H_I}
	\hat{H}_I = \sum_i \int \left(  g_{ k}^*  | 0 \rangle \langle i| \otimes b_{k,i}^{\dagger}+   g_{k}  | i \rangle \langle 0 | \otimes b_{k,i} \right) d k,
	\end{equation}
	i.e. each level interacts only with its own reservoir and the functions $  g_{k} $ (called the form factors \cite{Kozyrev17}) are the same for all reservoirs. Let us note that from the physical point of view such a Hamiltonian assumes that the dipole approximation (the only terms which are linear in creation and annihilation operators involved in the Hamiltonian) and the rotating wave approximation (the terms of the form $ | i \rangle \langle 0| \otimes b_{k,i}^{\dagger} $ are absent) are justified. The validity of the latter one in a general case is controversial \cite{Fleming10, Tang13}, but we do not discuss this question in our study.
	
	We consider the Schroedinger equation
	\begin{equation}\label{eq:Shr}
	\frac{d}{dt} | \Psi (t) \rangle = - i \hat{H} | \Psi (t) \rangle,
	\end{equation}
	with the Hamiltonian $ \hat{H} = \hat{H}_S \otimes I + I \otimes \hat{H}_B + \hat{H}_I $ and the initial condition
	\begin{equation}\label{eq:initCond}
	| \Psi (0) \rangle =( | \psi(0)\rangle + \psi_0(0) | 0 \rangle) \otimes | \Omega \rangle, \qquad \langle 0| \psi (0)\rangle = 0,
	\end{equation}
	i.e. we assume the initial condition to be completely factorized as in \cite{Teret19}. Thus, non-Markovian effects related to non-factorized initial state are also out of the range of our study. A quite general approach for such effects was suggested in \cite{Breuer07}. We also assume that the initial state of the system is pure and the reservoirs are in the vacuum states, i.e. in their Gibbs states at zero temperature.
	
	\section{Pseudomode approach and exact evolution}\label{sec:pseudomode}
	
	First of all let us show that 1-particle restriction of the Hamiltonian $ \hat{H} $ is related to a generalized Friedrichs model.  Let us introduce an injective map $ \hat{}: h \rightarrow \mathcal{H}, $  where $  h = \mathbb{C}^N  \oplus \bigoplus_{i=1}^N \mathcal{L}^2 (\mathbb{R})$ which for any $ | \psi_F  \rangle  \in h $ of the form
	\begin{equation*}
	| \psi_F \rangle = | \psi  \rangle \oplus \sum_{i=1}^N \int dk \psi_{k,i} |k, i \rangle, 
	\end{equation*}
	where $ \psi_{k,i} $ are from the Schwartz space $ \mathcal{S}(\mathbb{R}) $ and $ \int dk \cdot |k, i \rangle $ is a Fourier transform for each $ i=1, \ldots, m $, defines  $ | \hat{\psi}_F  \rangle \in \mathcal{H}  $ by the formula
	\begin{equation*}
	| \hat{\psi}_F  \rangle	= | \psi \rangle \otimes |\Omega\rangle +  | 0 \rangle \otimes \sum_{i=1}^{N} \int dk \psi_{k,i} b_{k,i}^{\dagger}| \Omega \rangle. 
	\end{equation*}
	Such a map we called one-particle second quantization in \cite{Teret19}, it is a special realization of the idea to consider non-composite systems as composite ones suggested in \cite{Chernega13, Chernega14, Chernega14a, Manko18}.
	
	\begin{theorem}
		The solution of the Cauchy problem \eqref{eq:Shr}-\eqref{eq:initCond} has the form
		\begin{equation}\label{eq:OnePartSol}
		| \Psi (t) \rangle = \psi_0(0) | 0 \rangle \otimes | \Omega \rangle + | \hat{\psi}_F (t) \rangle, 
		\end{equation}
		where $ | \psi_F(t) \rangle $ is a solution of the Cauchy problem for the Friedrichs model:
		\begin{equation}\label{eq:Fri}
		\frac{d}{dt}| \psi_F (t) \rangle = - i H_F | \psi_F (t) \rangle, \qquad | \psi_F (0) \rangle = | \psi (0) \rangle \oplus 0,
		\end{equation}
		where $ H_F = H_S \oplus H_B + H_I $,  $ H_S $ is defined by formula \eqref{eq:H_S} and
		\begin{equation*}
		H_B = \sum_{i=1}^N \int \omega_k  |k, i \rangle \langle k, i| d k, \qquad 	H_I = \sum_{i=1}^N \int \left(  g_{ k}^*  |k, i \rangle  \langle i|+   g_{k}  | i \rangle \langle k, i| \right) d k.
		\end{equation*}
	\end{theorem}
	The proof of this theorem is based on the direct substitution. But the deep reason for the preservation of 0-particle and 1-particle subspaces consists in the presence of the integral of motion 
	\begin{equation}\label{eq:IntOfM}
	\hat{N} = \sum_{i=1}^N | i \rangle \langle i | \otimes I + I \otimes \sum_{i=1}^N \int b_{k,i}^{\dagger} b_{k,i} d k,
	\end{equation}
	which is nothing else but the total number of particles in reservoir and excitations in the system.
	
	This N-level generalized Friedrichs model is close to the one considered in \cite{Antoniou03}, but there was only one reservoir coupled to all the excited states of the system. 
	
	Now we are going to obtain the reduced evolution but for a state vector rather than a density matrix as it is usually done for the master equation derivation. Let us define the projection $ P $ on the linear subspace $ \mathbb{C}^N $ in the Hilbert space $ h $ and the projection $ Q = I_N-P $ on the orthogonal complement to this subspace.
	
	\begin{theorem} Let the integral
		\begin{equation}\label{eq:Gt}
		G(t) = \int |g_{k}|^2 e^{-i \omega_{k} t} dk
		\end{equation}
		converge for all $ t \in \mathbb{R}_+ $ and define the continuous function $ G(t) $, then
		$ | \psi (t) \rangle  = P  | \psi_F(t) \rangle$ satisfies the integro-differential equation
		\begin{equation}\label{eq:integroDiff}
		\frac{d}{dt} | \psi (t) \rangle = - i H_S| \psi (t) \rangle - \int_{0}^{t} ds \; G(t-s) | \psi (s) \rangle
		\end{equation}
		with the initial condition $ | \psi (t) \rangle|_{t=0} = | \psi (0) \rangle$, where $ H_S $ is defined by formula \eqref{eq:H_S}.
	\end{theorem}
	
	\begin{proof}
		Let us prove this theorem by the projection approach to emphasize the closeness of this approach to the Nakajima-Zwanzig projection one.  Let us represent equation \eqref{eq:Fri} as the system
		\begin{equation*}
		\begin{cases}
		\frac{d}{dt} P| \psi_F (t) \rangle = -i P H_F P | \psi_F (t) \rangle -i P H_F Q | \psi_F (t) \rangle,\\
		\frac{d}{dt} Q| \psi_F (t) \rangle = -i Q H_F P | \psi_F (t) \rangle -i Q H_F Q | \psi_F (t) \rangle.
		\end{cases}
		\end{equation*}
		Let us solve the second equation as a linear differential with respect to $ Q| \psi_F (t) \rangle  $ considering  $ -i Q H_F P | \psi_F (t) \rangle $ as an inhomogeneity
		\begin{equation*}
		Q| \psi_F (t) \rangle =e^{-i Q H_F Q t} Q| \psi_F (0)  \rangle -i \int_0^t ds e^{-i Q H_F Q (t-s)} Q H_F P | \psi_F (s) \rangle.
		\end{equation*}
		Substituting into the first one we obtain
		\begin{align}\label{eq:integroDiffGene}
		\frac{d}{dt} P| \psi_F (t) \rangle = -i P H_F P | \psi_F (t) \rangle &-i P H_F Q e^{-i Q H_F Q t} Q| \psi_F (0)  \rangle - \nonumber\\
		&- P H_F Q  \int_0^t ds \; e^{-i Q H_F Q (t-s)} Q H_F P | \psi_F (s) \rangle.
		\end{align}
		Taking into account
		\begin{gather*}
		P| \psi_F (t) \rangle = | \psi (t) \rangle, \quad P H_F P = H_S, \quad  Q| \psi_F (0)  \rangle = 0, \quad QHQ = H_B,\\
		Q H_F P = \sum_{i=1}^N \int  g_{ k}^*  |k, i \rangle  \langle i| d k, \qquad P H_F Q = \sum_{i=1}^N   g_{k}  | i \rangle \langle k, i| d k,
		\end{gather*}
		we have
		\begin{gather*}
		P H_F Q  e^{-i Q H_F Q t} Q H_F P =  \sum_{i,j=1}^N \int dk \; dk'  g_{k} g_{k'}^*  | i \rangle \langle k, i| e^{-i H_B t} |k', j \rangle  \langle j| = \\
		= \sum_{i,j=1}^N \int dk \; dk'  g_{k} g_{k'}^*  | i \rangle e^{- i \omega_{k} t} \delta_{ij} \delta(k-k')  \langle j| = I_{N}\int dk e^{-i \omega_k t} |g_k|^2 = G(t)I_N.
		\end{gather*}
		Substituting into \eqref{eq:integroDiffGene} we obtain \eqref{eq:integroDiff}.
		
	\end{proof}
	
	Let us note that the solution of equation \eqref{eq:integroDiff} with the initial condition $ | \psi (t) \rangle|_{t=0} = | \psi (0) \rangle$ and a continuous function $ G(t) $ exists and is unique \cite[Sec. 2.1]{Burton05}
	
	Thus, equation \eqref{eq:integroDiff} is an analog of the Nakajima-Zwanzig equation but for the state vector. Such approach is called Feshbach projection approach \cite{Feshbach58, Feshbach62}. Its application to open systems and its generalization describing completely positive evolution of a reduced density matrix could be found in \cite{Chruscinski13}.
	
	In fact, we have already used the pseudomode approach in the interaction picture in \cite{Teret19}, but we have not dwellt on it explicitly. So let us introduce here the equation in the interaction picture in the explicit form.
	\begin{corollary}
		The vector $ | \psi_I (t) \rangle \equiv e^{i H_S t} | \psi (t) \rangle$ satisfies the equation
		\begin{equation}\label{eq:integroDiffInt}
		\frac{d}{dt} | \psi_I (t) \rangle = - \int_{0}^{t} ds \; G(t-s) e^{i H_S (t-s)} | \psi_I (s) \rangle
		\end{equation}
		with the initial condition$ | \psi_I (t) \rangle|_{t=0} = | \psi (0) \rangle$.
	\end{corollary}
	
	It is important that $ | \psi (t) \rangle $ and $ | \psi_I (t) \rangle $ unambiguously define the evolution of a reduced density matrix or of a reduced density matrix in the interaction picture, accordingly.
	
	\begin{lemma}
		The reduced density matrix 
		\begin{equation*}
		\rho_S(t) \equiv \mathrm{tr}_R \; | \Psi (t) \rangle \langle \Psi (t)|,
		\end{equation*}
		where $ \mathrm{tr}_R \; $ is a partial trace with respect to the space $ \bigotimes\limits_{i=1}^N \mathfrak{F}_b(\mathcal{L}^2(\mathbb{R})) $ and $ | \Psi (t) \rangle  $ is defined by formula \eqref{eq:OnePartSol}, has the form
		\begin{equation*}
		\rho_S(t) = 0 \oplus|\psi(t)\rangle \langle \psi(t)| +  \psi_0(0)^*|\psi(t)\rangle \langle 0|  + \psi_0(0) |0\rangle \langle \psi(t)| + (1 - || \psi(t)||^2) |0\rangle \langle 0|
		\end{equation*}
		as well as the reduced density matrix in the interaction picture $ \rho_{SI}(t) \equiv e^{i \hat{H}_S t} \rho_S(t) e^{-i \hat{H}_S t}  $ has the form
		\begin{equation}\label{eq:rhoSI}
		\rho_{SI}(t) =  0 \oplus|\psi_I(t)\rangle \langle \psi_I(t)| +  \psi_0^*(0)|\psi_I(t)\rangle \langle 0|  + \psi_0(0) |0\rangle \langle \psi_I(t)| + (1 - || \psi_I(t)||^2) |0\rangle \langle 0|.
		\end{equation}
	\end{lemma}
	
	This lemma could be proved by direct calculation (see also discussion in \cite[Sec. 4]{Teret19}).
	
	In this paper we focus on the case, when $ G(t) $ has the special form 
	\begin{equation}\label{eq:ourGt}
	G(t) = g^2 e^{- \frac{\gamma}{2} |t| - i \varepsilon t}, \qquad  \gamma >0, g>0, \varepsilon \in \mathbb{R},
	\end{equation}
	although the approach presented below could be immediately generalized to the case of the combination of exponentials $ G(t) = \sum_j g_j^2 e^{- \frac{\gamma_j}{2} |t| - i \varepsilon_j t} $ (see \cite{Teret19}). In physics usually the spectral density $ \mathcal{J}(\omega) $ \cite{Garraway97} defined by 
	\begin{equation*}
	G(t) = \int_{- \infty}^{+\infty} \frac{d \omega}{2 \pi} e^{- i \omega t}  \mathcal{J}(\omega), \qquad  \mathcal{J}(\omega) = \int_{- \infty}^{+\infty} G(t)  e^{i \omega t} dt
	\end{equation*}
	rather than the function $ G(t) $  is given from the experimental data. In our case, when $ G(t) $ has form \eqref{eq:ourGt}, we have
	\begin{equation}\label{eq:spectralDensity}
	\mathcal{J}(\omega) = \frac{\gamma g^2}{ \left(\frac{\gamma}{2}\right)^2 + (\omega -\varepsilon)^2},
	\end{equation}
	i.e. nothing else but Lorentzian spectral density.
	\begin{theorem}\label{th:nonHermEq}
		Let $ | \psi (t) \rangle $ be a solution of \eqref{eq:integroDiff} with the initial condition $ | \psi (t) \rangle|_{t=0} = | \psi (0) \rangle$ in the case, when $ G(t) $ has form \eqref{eq:ourGt}, then  $ | \tilde{\psi} (t) \rangle \equiv | \psi (t) \rangle \oplus | \varphi (t) \rangle \in \mathbb{C}^N \oplus \mathbb{C}^N $, where 
		\begin{equation}\label{eq:pseudoDef}
		| \varphi (t) \rangle \equiv -i g \int_{0}^t ds \; e^{- \left(\frac{\gamma}{2} + i \varepsilon \right) (t-s)} | \psi (s) \rangle,
		\end{equation}
		satisfies the Schroedinger equation with a non-Hermitian (dissipative) Hamiltonian
		\begin{equation}\label{eq:nonHermEq}
		\frac{d}{dt} | \tilde{\psi} (t) \rangle = - i H_{\rm eff} | \tilde{\psi} (t) \rangle, \qquad H_{\rm eff} = 
		\begin{pmatrix}
		H_S & g I_N\\
		g I_N & \left(\varepsilon - i\frac{\gamma}{2} \right) I_N
		\end{pmatrix}
		\end{equation}
		with the initial condition $ | \tilde{\psi} (t) \rangle  = | \psi (0) \rangle  \oplus 0 $.
	\end{theorem}
	
	\begin{proof}
		Substituting \eqref{eq:ourGt} in \eqref{eq:integroDiff} and taking into account \eqref{eq:pseudoDef} we obtain
		\begin{equation*}
		\frac{d}{dt}| \psi (t) \rangle = - i H_S | \psi (t) \rangle - i g | \varphi (t) \rangle
		\end{equation*} 
		and differentiating \eqref{eq:pseudoDef} with respect to $ t $ we have
		\begin{equation*}
		\frac{d}{dt} | \varphi (t) \rangle = - i g | \psi (t) \rangle - \left(\frac{\gamma}{2} + i \varepsilon \right) | \varphi (t) \rangle.
		\end{equation*}
		By combining these two ordinary differential equations into one for the vector $ | \tilde{\psi} (t) \rangle = | \psi (t) \rangle \oplus | \varphi (t) \rangle $ we obtain \eqref{eq:nonHermEq}.  
	\end{proof}
	
	Actually, the non-Hermitian part of $ H_{\rm eff} $ is very close to the optical potential concept from the nuclear physics \cite{Fermi54, Feshbach54}.
	
	\begin{corollary}
		The vector $ | \tilde{\psi}_I (t) \rangle \equiv | \psi_I (t) \rangle \oplus e^{i H_S t} | \varphi (t) \rangle $ satisfies the equation
		\begin{equation}\label{eq:nonHermEqInt}
		\frac{d}{dt} | \tilde{\psi} (t) \rangle = - i H_{I, \rm eff} | \tilde{\psi} (t) \rangle, \qquad H_{I, \rm eff} = 
		\begin{pmatrix}
		0 & g I_N\\
		g I_N & -H_S + \left(\varepsilon - i\frac{\gamma}{2} \right) I_N
		\end{pmatrix}
		\end{equation}
		with the initial condition $ | \tilde{\psi}_I (t) \rangle  = | \psi (0) \rangle  \oplus 0 $.
	\end{corollary}
	Equation \eqref{eq:nonHermEqInt} could be solved in the global basis, i.e. the eigenbasis of $ H_S $ \cite{Levy14, Trushechkin16}. We will numerate the global basis by the Greek letters: 
	\begin{equation*}
	| i \rangle =\sum_{\alpha=1}^{N} U_{i\alpha}| \alpha \rangle, \qquad H_S = \sum_{\alpha=1}^N E_{\alpha} | \alpha \rangle \langle \alpha| .
	\end{equation*}
	We also define the energy detunings
	\begin{equation*}
	\Delta E_{\alpha} \equiv  E_{\alpha} - \varepsilon.
	\end{equation*}
	\begin{corollary}
		Let us decompose $ | \psi_I (t) \rangle = \sum_{\alpha=1}^N \psi_{\alpha}(t) | \alpha \rangle ,  \oplus e^{i H_S t} | \varphi(t) \rangle = \sum_{\alpha=1}^N \varphi_{\alpha}(t) | \alpha \rangle $, then the coefficients of the decomposition satisfy the systems of equations
		\begin{equation*}
		\frac{d}{dt} \begin{pmatrix}
		\psi_{\alpha}(t) \\
		\varphi_{\alpha}(t)
		\end{pmatrix}
		= \begin{pmatrix}
		0  & - i g\\
		- i g & i \Delta E_{\alpha} - \frac{\gamma}{2}
		\end{pmatrix}
		\begin{pmatrix}
		\psi_{\alpha}(t) \\
		\varphi_{\alpha}(t)
		\end{pmatrix}, \qquad \alpha =1, \ldots, N.
		\end{equation*}
		The eigenvalues of the matrix $ - i H_{I, \rm eff} $ equal
		\begin{equation*}
		\lambda_{\alpha, \pm}^{\rm exact} =-\frac12 \left(\frac{\gamma}{2} - i \Delta E_{\alpha} \right) \pm \frac12 \sqrt{\left(\frac{\gamma}{2}\right)^2 - \Delta E_{\alpha}^2 + 4 g^2 - i \gamma \Delta E_{\alpha}}, \qquad \alpha =1, \ldots, N. 
		\end{equation*}
	\end{corollary}
	
	We are interested in the population decay rates and the decoherence rates. They are defined by the real parts of the eigenvalues
	\begin{equation}\label{eq:Relambdsalpha}
	\mathrm{Re}\; \lambda_{\alpha, \pm}^{\rm exact}  = - \frac{\gamma}{4} \pm \frac{\sqrt{2}}{4} \sqrt{\left(\frac{\gamma}{2}\right)^2 -4 g^2 - \Delta E_{\alpha}^2 + \sqrt{\left( \left(\frac{\gamma}{2} +2 g\right)^2 +\Delta E_{\alpha}^2\right) \left( \left(\frac{\gamma}{2} -2 g\right)^2 + \Delta E_{\alpha}^2\right)}}.
	\end{equation}
	The population decay rates in the global basis are characterized by the minimal absolute value of these eigenvalues 
	\begin{equation}\label{eq:etaaEx}
	\eta_{\alpha}^{\rm exact} =  -2\mathrm{Re}\; \lambda_{\alpha, +}^{\rm exact}, 
	\end{equation}
	as well as the decoherence rates are characterized by
	\begin{equation}\label{eq:etaabEx}
	\eta_{\alpha \beta}^{\rm exact} =  -\mathrm{Re}\; \lambda_{\alpha, +}^{\rm exact} - \mathrm{Re}\; \lambda_{\beta, +}^{\rm exact}, \qquad \eta_{\alpha 0}^{\rm exact} = -\mathrm{Re}\; \lambda_{\alpha, +}^{\rm exact} .
	\end{equation}
	
	At the end of this section let us concentrate on the fact that equation \eqref{eq:nonHermEq} is equivalent to a GKSL equation for a density matrix in the $ (2N+1) $-dimensional space.
	\begin{theorem}
		The $ (2N+1) $ by $ (2N+1) $ matrix
		\begin{equation*}
		\tilde{\rho}(t) \equiv 0 \oplus | \tilde{\psi} (t) \rangle \langle \tilde{\psi} (t) | + \overline{\psi}_0(0)| 0 \rangle \langle \tilde{\psi} (t) | +  \psi_0(0)| \tilde{\psi} (t) \rangle \langle 0 | + (1 -||\tilde{\psi} (t) ||^2)| 0 \rangle \langle 0 | \in \mathbb{C}^{(2N+1) \times (2N+1)}
		\end{equation*}
		satisfies the GKSL equation 
		\begin{equation*}
		\frac{d}{dt} \tilde{\rho}(t) = - i [H, \tilde{\rho}(t)] + \sum_{l=1}^N \left(L_l \tilde{\rho} L_l^{\dagger} -\frac12 L_l^{\dagger} L_l \tilde{\rho} -\frac12 \tilde{\rho}  L_l^{\dagger} L_l \right),  H = 0\; \oplus \frac12(H_{\rm eff}^{\dagger}+H_{\rm eff}),  L_l = |0 \rangle \langle \tilde{l}|,
		\end{equation*}
		where $ |0\rangle, |1\rangle, \ldots, |N\rangle, |\tilde{1}\rangle, \ldots, |\tilde{N}\rangle $ is the basis of $  \mathbb{C}^{2N+1} $.
	\end{theorem}
	The proof of this theorem could be found in \cite[Proposition 1]{Teret19}. Thus, we have built up the GKSL description of the non-Markovian evolution of the reduced density matrix at the cost of introduction of additional degrees of freedom (see recent discussion in \cite{Teret19}, \cite{Luchnikov19}). This GKSL equation (similar to the Friedrichs model) could be considered as a one-particle restriction of models in larger Hilbert spaces. For example, multi-level generalizations of the Jaynes-Cummings model with dissipation \cite{Sachdev84} could be considered in such a way. 
	
	\section{Nakajima-Zwanzig equation in Born approximation}\label{sec:NakZwanBorn}
	
	The Nakajima-Zwanzig equation in the Born approximation \cite[p. 131]{Breuer02}, \cite[p.9]{Weiss99} at zero temperature (the state of the reservoir is $  | \Omega \rangle \langle \Omega | $)
	\begin{equation}\label{eq:BornApprNZ}
	\frac{d}{dt} \rho_{SI}(t) = - \int_0^t ds  \mathrm{Tr}_B [\hat{H}_I(t), [\hat{H}_I(s), \rho_{SI}(s) \otimes | \Omega \rangle \langle \Omega |]],
	\end{equation}
	where $ \hat{H}_I(t) $ is the iteration Hamiltonian in the interaction picture, i.e.
	\begin{equation}\label{eq:HIntIntGen}
	\hat{H}_I(t)\equiv e^{i \hat{H}_S \otimes \hat{H}_B} \hat{H}_I e^{-i \hat{H}_S \otimes \hat{H}_B}.
	\end{equation}
	In our case $ \hat{H}_S, \hat{H}_B, \hat{H}_I  $ are defined by formulae \eqref{eq:H_S}, \eqref{eq:H_B} and \eqref{eq:H_I}. This equation could be obtained in the second order of the perturbation theory with respect to coupling constant between a system and a reservoir \cite[Subsec. 9.1.1]{Breuer02}, \cite{Trushechkin19}.
	\begin{lemma}\label{lem:HIntInt}
		The interaction Hamiltonian \eqref{eq:HIntIntGen} has the form
		\begin{equation}\label{eq:HIntInt}
		\hat{H}_I(t) =  \sum_{i, \alpha}  \left( | 0 \rangle \langle \alpha| \otimes b_{\alpha, i}^{\dagger}(t)+   | \alpha \rangle \langle 0 | \otimes b_{\alpha, i}(t) \right), 
		\end{equation}
		where
		\begin{equation}\label{eq:bt}
		b_{\alpha, i}(t) \equiv \int U_{i\alpha} g_{k}  e^{-i (\omega_k - E_{\alpha}) t} b_{k,i} dk.
		\end{equation}
		Moreover, $ b_{\alpha, i}(t) | \Omega \rangle = 0$ and the following commutation relations are held
		\begin{equation}\label{eq:commRelInt}
		[b_{\alpha, i}(t), b_{\beta, j}^{\dagger}(s)] =
		\delta_{ij} U_{i\alpha}  U_{i\beta}^* G(t-s)  e^{i( E_{\alpha} t - E_{\beta} s) }, \qquad [b_{\alpha, i}(t), b_{\beta, j}(s)]=0,
		\end{equation}
		where $ G(t) $ is defined by \eqref{eq:Gt}.
	\end{lemma}
	
	\begin{proof}
		Let us represent the interaction Hamiltonian in the global basis
		\begin{equation*}
		\hat{H}_I = \sum_{i, \alpha} \int \left(  g_{ k}^* U_{i\alpha}^*  | 0 \rangle \langle \alpha| \otimes b_{k,i}^{\dagger}+  U_{i\alpha} g_{k}  | \alpha \rangle \langle 0 | \otimes b_{k,i} \right) d k, \qquad i, \alpha =1 ,\ldots,N.
		\end{equation*}
		Taking into account $ e^{i \hat{H}_S t}| \alpha \rangle \langle 0 | e^{-i \hat{H}_S t} = e^{i E_{\alpha} t} | \alpha \rangle \langle 0 | $ and $e^{i \hat{H}_B t} b_{k,i} e^{-i \hat{H}_B t} = e^{-i \omega_{k} t} b_{k,i} $ we obtain
		\begin{equation*}
		\hat{H}_I(t) =  \sum_{i, \alpha} \int \left(  g_{ k}^* U_{i\alpha}^*  e^{i (\omega_k - E_{\alpha}) t} | 0 \rangle \langle \alpha| \otimes b_{k,i}^{\dagger}+  U_{i\alpha} g_{k}  e^{-i (\omega_k - E_{\alpha}) t} | \alpha \rangle \langle 0 | \otimes b_{k,i} \right) d k.
		\end{equation*}
		Thus, we obtain \eqref{eq:HIntInt}. $ b_{\alpha, i}(t) | \Omega \rangle = 0$ and the second of commutation relations \eqref{eq:commRelInt} follows from definition \eqref{eq:bt} of $ b_{\alpha, i}(t) $ which is a linear combination of annihilation operators. Let us calculate
		\begin{align*}
		[b_{\alpha, i}(t), b_{\beta, j}^{\dagger}(s)] &= \left[\int U_{i\alpha} g_{k}  e^{-i (\omega_k - E_{\alpha}) t} b_{k,i} dk, \int U_{j\beta}^* g_{p}^* e^{-i (\omega_p - E_{\beta}) s} b_{p,j}^{\dagger} dp\right] =\\
		&= \int \int dk dp \; U_{i\alpha} g_{k}  e^{-i (\omega_k - E_{\alpha}) t}  U_{j\beta}^* g_{p}^* e^{-i (\omega_p - E_{\beta}) s} [b_{k,i}, b_{p,j}^{\dagger}] = \\
		&= \int \int dk dp \; U_{i\alpha} g_{k}  e^{-i (\omega_k - E_{\alpha}) t}  U_{j\beta}^* g_{p}^* e^{i (\omega_p - E_{\beta}) s} \delta_{ij} \delta(k-p)=\\
		&= \delta_{ij} \int  dk  \; U_{i\alpha}  U_{i\beta}^* |g_{k}|^2  e^{-i (\omega_k - E_{\alpha}) t}  e^{i (\omega_k - E_{\beta}) s}  =
		\delta_{ij} U_{i\alpha}  U_{i\beta}^* G(t-s)  e^{i( E_{\alpha} t - E_{\beta} s) } .
		\end{align*}
		Thus, we obtain the first  of commutation relations \eqref{eq:commRelInt}.
	\end{proof}

	\begin{lemma}\label{lem:ourBornApprNZ}
		In our case, i.e., when $ \hat{H}_I(t) $ is defined by formula \eqref{eq:HIntInt}, equation \eqref{eq:BornApprNZ} takes the form
		\begin{align}\label{eq:ourBornApprNZ}
		\frac{d}{dt} \rho_{SI}(t)  &= \int_0^t ds \; G(t-s) \left(| 0 \rangle \langle 0 |   \mathrm{Tr}\; (\Pi  e^{i \hat{H}_S (t-s)} \rho_{SI}(s) ) - \Pi  e^{i \hat{H}_S (t-s)} \rho_{SI}(s)  \right)  +  \nonumber \\
		&+ \int_0^t ds \; G^*(t-s) \left( | 0 \rangle \langle 0 |  \mathrm{Tr}\; (\rho_{SI}(s)  e^{-i \hat{H}_S (t-s)} \Pi)    - \rho_{SI}(s) e^{-i \hat{H}_S (t-s)}  \Pi \right),
		\end{align}
		where the projection
		\begin{equation}\label{eq:projPi}
		\Pi = \sum_{i=1}^N| i \rangle \langle i | = I_{N+1} -| 0 \rangle \langle 0 | = 0 \oplus I_N
		\end{equation}
		is introduced.
	\end{lemma}
	
	\begin{proof}
		Let us expand the following expression
		\begin{align}\label{eq:doubleComm}
		& \mathrm{Tr}_B [H_I(t), [H_I(s), \rho_{SI}(s) \otimes |\Omega \rangle \langle \Omega |]] = \nonumber\\
		&=  \langle \Omega |  H_I(t) H_I(s) |\Omega \rangle  \rho_{SI}(s) + \rho_{SI}(s) \langle \Omega |  H_I(s) H_I(t) |\Omega \rangle  - \nonumber\\
		&-   \mathrm{Tr}_B H_I(t)  |\Omega \rangle \rho_{SI}(s)  \langle \Omega | H_I(s) -   \mathrm{Tr}_B H_I(s)  |\Omega \rangle \rho_{SI}(s)  \langle \Omega | H_I(t)
		\end{align}
		Thus, only two terms could be calculated, the other two could be obtained by interchange of $ t  $ and $ s $. By lemma \ref{lem:HIntInt} we obtain
		\begin{align*}
		\langle \Omega |  H_I(t) H_I(s) |\Omega \rangle \rho_{SI}(s)  &= \sum_{i, \alpha, \beta} | \alpha \rangle \langle \beta |  U_{i\alpha}  U_{i\beta}^* G(t-s)  e^{i( E_{\alpha} t - E_{\beta} s) } \rho_{SI}(s) =\\
		&=  e^{i \hat{H}_S t} \sum_{i, \alpha, \beta} | \alpha \rangle \langle \beta |  U_{i\alpha}  U_{i\beta}^* G(t-s)  e^{-i \hat{H}_S s} \rho_{SI}(s) \\
		&=  e^{i \hat{H}_S t} \sum_{i} | i \rangle \langle i | G(t-s)  e^{-i \hat{H}_S s} \rho_{SI}(s) =  \Pi  e^{i \hat{H}_S (t-s)} \rho_{SI}(s)
		\end{align*}
		and
		\begin{align*}
		& \mathrm{Tr}_B H_I(t)  |\Omega \rangle \rho_{SI}(s)  \langle \Omega | H_I(s) = \\
		&= \mathrm{Tr}_B \sum_{ij, \alpha \beta}  \left( | 0 \rangle \langle \alpha| \otimes b_{\alpha, i}^{\dagger}(t)+   | \alpha \rangle \langle 0 | \otimes b_{\alpha, i}(t) \right) |\Omega \rangle \rho_{SI}(s)  \langle \Omega |\left( | 0 \rangle \langle \beta| \otimes b_{\beta, j}^{\dagger}(s)+   | \beta \rangle \langle 0 | \otimes b_{\beta, j}(s) \right) = \\
		&=  \mathrm{Tr}_B \sum_{ij, \alpha \beta}   | 0 \rangle \langle \alpha| \otimes b_{\alpha, i}^{\dagger}(t)|\Omega \rangle \rho_{SI}(s)  \langle \Omega |  | \beta \rangle \langle 0 | \otimes b_{\beta, j}(s)  =\\
		&= \sum_{i, \alpha, \beta}  | 0 \rangle \langle \alpha| \rho_{SI}(s) | \beta \rangle \langle 0 | U_{i\alpha}^*  U_{i\beta} G(s-t)  e^{-i( E_{\alpha} t - E_{\beta} s) } =\\
		&=  \sum_{i, \alpha, \beta} U_{i\alpha}^*  U_{i\beta}  | 0 \rangle \langle \alpha| e^{-i \hat{H}_S t} \rho(s) e^{i \hat{H}_S s} | \beta \rangle \langle 0 |  G(s-t) =\\
		&= \sum_{i}    | 0 \rangle \langle i| e^{-i \hat{H}_S t} \rho(s) e^{i \hat{H}_S s} | i \rangle \langle 0 |  G(s-t) =\\
		&=G^*(t-s) | 0 \rangle \langle 0 |  \mathrm{Tr}\; (\Pi e^{-i \hat{H}_S t} \rho(s) e^{i \hat{H}_S s})   =G^*(t-s) | 0 \rangle \langle 0 |  \mathrm{Tr}\; (\Pi e^{-i \hat{H}_S (t-s)} \rho(s)) ,
		\end{align*}
		where $ [\Pi, \hat{H}_S]  = [0 \oplus I_N, 0 \oplus H_S] = 0 $ is used. By substituting the obtained expressions into \eqref{eq:doubleComm} and then into \eqref{eq:BornApprNZ} we obtain \eqref{eq:ourBornApprNZ}.
	\end{proof}
	
	Equation \eqref{eq:ourBornApprNZ} has a convolution form  and, hence, could be solved by means of Laplace transform \cite[p. 30]{Burton05}. These convolution form is preserved in the Schroedinger picture:
	\begin{align*}
	\frac{d}{dt} \rho_S(t) = -i [\hat{H}_S, \rho(t) ] &+ \int_0^t ds G(t-s) \left(| 0 \rangle \langle 0 |   \mathrm{Tr}\; (\Pi \rho_S(s) e^{i \hat{H}_S (t-s)} ) - \Pi  \rho_S(s) e^{i \hat{H}_S (t-s)} \right)  +   \\
	&+ \int_0^t ds G^*(t-s) \left( | 0 \rangle \langle 0 |  \mathrm{Tr}\; ( e^{-i \hat{H}_S (t-s)}  \rho_S(s) \Pi)    - e^{-i \hat{H}_S (t-s)} \rho_S(s)\Pi \right).
	\end{align*}
	
	\begin{lemma}\label{lem:BornApprNZsol}
		The solution of equation \eqref{eq:ourBornApprNZ} could be represented in the form
		\begin{equation}\label{eq:formOfSol}
		\rho_{SI}(t) = \rho_{00}(t)| 0\rangle \langle 0 | + \overline{\psi}_0(0) | \psi_I(t) \rangle \langle 0 | + \psi_0(0) | 0 \rangle \langle \psi_I(t) | + 0 \oplus \sigma(t), \qquad \sigma(t) = \sum_{i , j=1}^N \rho_{ij}(t) | i \rangle \langle j |, 
		\end{equation}
		where the vector-valued function $ | \psi_I(t) \rangle $ satisfies equitation  \eqref{eq:integroDiffInt} with the initial condition $ | \psi_I (t) \rangle|_{t=0} = | \psi (0) \rangle$, the matrix-valued function $ \sigma(t) $ satisfies the equitation  
		\begin{equation}\label{eq:sigma}
		\frac{d}{dt} \sigma(t) = - \int_0^t ds G(t-s) e^{i H_S (t-s)} \sigma(s)    - \int_0^t ds G^*(t-s) \sigma(s) e^{-i H_S (t-s)}
		\end{equation}
		and the scalar function $ \rho_{00}(t) $ could be defined from the normalization condition $ \rho_{00}(t) = 1-  \mathrm{Tr}\; \sigma(t) $.
	\end{lemma}
	
	Thus, without additional assumptions about function $ G(t) $ we obtain that the coherences between excited and ground state of the system  (defined by the vector $  | \psi_I(t) \rangle  $) coincide for exact \eqref{eq:rhoSI} and approximate \eqref{eq:formOfSol} solutions. 
	
	Now let us consider the special case, when the function $ G(t) $ is defined by formula \eqref{eq:ourGt}. Similar to the solution of equation \eqref{eq:integroDiffInt} for $ | \psi_I (t) \rangle $ (see theorem \ref{th:nonHermEq}) one could use a similar technique to solve  equation \eqref{eq:sigma}. This technique is also very close to auxiliary density
	matrices method developed in \cite{Meier99, Pomyalov10}.
	
	\begin{lemma}
		Let $ \sigma(t) $ satisfy equation \eqref{eq:sigma}, $ G(t) $ have form \eqref{eq:ourGt} and 
		\begin{equation*}
		X(t) \equiv  - i g \int_0^t ds \;  e^{- \left(\frac{\gamma}{2}   +i \varepsilon - i H_S\right)(t-s)}  \sigma(s), 
		\end{equation*}
		then $ \sigma(t) $ and $ X(t) $ satisfy the system of linear matrix equations with constant coefficients
		\begin{equation}\label{eq:sigmaSystem}
		\begin{cases}
		\frac{d}{dt} \sigma(t) &=- ig X(t) + i gX^{\dagger}(t),\\
		\frac{d}{dt} X(t) &=- \left(\frac{\gamma}{2} - i( H_S - \varepsilon)\right) X(t)- i g \sigma(t).
		\end{cases}
		\end{equation}
	\end{lemma}
	
	As well as for the case of equations \eqref{eq:nonHermEqInt} it is natural to solve these system in the global basis.
	
	\begin{corollary}
		If one decomposes $ \sigma(t) $ and $ X(t) $ in the global basis $ \sigma(t) = \sum_{\alpha\beta}\sigma_{\alpha \beta}(t) | \alpha \rangle \langle \beta | $, $ X(t) = \sum_{\alpha\beta}X_{\alpha \beta}(t) | \alpha \rangle \langle \beta | $, then the coefficients of the decomposition satisfy the systems
		\begin{equation*}
		\frac{d}{dt} \begin{pmatrix}
		\sigma_{\alpha \beta}(t)\\
		X_{\alpha \beta}(t)\\
		- X_{\beta \alpha }^*(t)
		\end{pmatrix} =
		\begin{pmatrix}
		0 & - i g & -ig\\
		-ig & - \frac{\gamma}{2}   - i \Delta E_{\alpha}   & 0\\
		-ig & 0 & - \frac{\gamma}{2}   + i \Delta E_{\beta}  \\
		\end{pmatrix}
		\begin{pmatrix}
		\sigma_{\alpha \beta}(t)\\
		X_{\alpha \beta}(t)\\
		- X_{\beta \alpha }^*(t)
		\end{pmatrix}, \quad \alpha, \beta =1, \ldots, N.
		\end{equation*}
	\end{corollary}
	So the solution of system \eqref{eq:sigmaSystem} is reduced to the calculation of $ 3 \times 3 $-matrix exponential. Hence, the evolution is defined by the eigenvalues of such a matrix which could be expressed as zeros of the characteristic polynomial 
	\begin{align*}
	f_{\alpha \beta}(\lambda) = \lambda^3 &+ (\gamma +   i(\Delta E_{\beta} - \Delta E_{\alpha}))\lambda^2 + \left( \left(\frac{\gamma}{2}\right)^2 + i \frac{\gamma}{2} (\Delta E_{\beta} - \Delta E_{\alpha}) +2 g^2 + \Delta E_{\alpha}  \Delta E_{\beta} \right) \lambda +\\
	&+ g^2 (\gamma +  i (\Delta E_{\beta} - \Delta E_{\alpha})).
	\end{align*}
	Namely,
	\begin{equation}\label{eq:etaabB}
	\eta_{\alpha \beta} \equiv \min_{f_{\alpha \beta}(\lambda)=0} | \mathrm{Re} \; \lambda|.
	\end{equation}
	In particular the population decay rate in the global basis is defined by zeros of the characteristic polynomial 
	\begin{equation}\label{eq:falpha}
	f_{\alpha}(\lambda) \equiv f_{\alpha\alpha}(\lambda) = \lambda^3 + \gamma \lambda^2 + \left( \left(\frac{\gamma}{2}\right)^2 +2 g^2 +\Delta E_{\alpha}^2\right) \lambda + g^2 \gamma
	\end{equation}
	with real coefficients. Namely,
	\begin{equation}\label{eq:etaaB}
	\eta_{\alpha} \equiv \eta_{\alpha \alpha} = \min_{f_{\alpha}(\lambda)=0} | \mathrm{Re} \; \lambda|.
	\end{equation}
	
	\section{Redfield equation}\label{sec:Redfield}
	
	In physical literature \cite[p. 132]{Breuer02}, \cite[p. 141]{May08} the following equation is frequently used:
	\begin{equation}\label{eq:Red}
	\frac{d}{dt} \rho_{SI}(t) = - \int_0^t ds  \mathrm{Tr}_B [\hat{H}_I(t), [\hat{H}_I(s), \rho_{SI}(t) \otimes | \Omega \rangle \langle \Omega |]]
	\end{equation}
	which is called the (non-Markovian) Redfield master equation \cite{Redfield65}. In \eqref{eq:Red} we have taken into account zero temperature of the reservoirs. This equation differs from equation  \eqref{eq:BornApprNZ} by the fact that $ \rho_{SI} $ inside the integral is a function of $ t $ rather than $ s $. Thus,  equation \eqref{eq:Red} has a time-local generator but this generator is time-dependent. Moreover, at $ t = 0 $ the integral in the left hand side of \eqref{eq:BornApprNZ} vanishes, which leads to zero population and coherence decay rates at $ t=0 $.
	
	After the full Markovian approximation \cite[p. 132]{Breuer02}, \cite[p. 145]{May08} Redfield equation \eqref{eq:Red} takes the form
	\begin{equation}\label{eq:GKSL}
	\frac{d}{dt} \rho_{SI}(t) = - \int_0^{+\infty} ds  \mathrm{Tr}_B [\hat{H}_I(t), [\hat{H}_I(t -s), \rho_{SI}(t) \otimes | \Omega \rangle \langle \Omega |]]
	\end{equation}
	which we call the Markovian Redfield master equation.
	
	\begin{lemma}
		In our case, i.e., when $ \hat{H}_I(t) $ is defined by \eqref{eq:HIntInt}, equation \eqref{eq:Red} takes the form
		\begin{align}\label{eq:ourRedfield}
		\frac{d}{dt} \rho_{SI}(t)  &= \int_0^t ds \; G(t-s) \left(| 0 \rangle \langle 0 |   \mathrm{Tr}\; (\Pi  e^{i \hat{H}_S (t-s)} \rho_{SI}(t) ) - \Pi  e^{i \hat{H}_S (t-s)} \rho_{SI}(t)  \right)  +  \nonumber \\
		&+ \int_0^t ds \; G^*(t-s) \left( | 0 \rangle \langle 0 |  \mathrm{Tr}\; (  \rho_{SI}(t)  e^{-i \hat{H}_S (t-s)} \Pi)    - \rho_{SI}(t) e^{-i \hat{H}_S (t-s)}  \Pi \right)
		\end{align}
		and equation \eqref{eq:GKSL}  takes the form
		\begin{align}\label{eq:ourGKSL}
		\frac{d}{dt} \rho_{SI}(t)  &= \int_0^{+ \infty} ds \; G(s) \left(| 0 \rangle \langle 0 |   \mathrm{Tr}\; (\Pi  e^{i \hat{H}_S s} \rho_{SI}(t) ) - \Pi  e^{i \hat{H}_S s} \rho_{SI}(t)  \right)  +  \nonumber \\
		&+ \int_0^{+ \infty} ds \; G^*(s) \left( | 0 \rangle \langle 0 |  \mathrm{Tr}\; (  \rho_{SI}(t)  e^{-i \hat{H}_S s} \Pi)    - \rho_{SI}(t) e^{-i \hat{H}_S s}  \Pi \right),
		\end{align}
		where the projection $ \Pi $ is defined by formula \eqref{eq:projPi}.
	\end{lemma}
	
	The proof of this lemma is analogous to the proof of lemma \ref{lem:ourBornApprNZ}. 
	
	\begin{corollary}
		Let $ G(t) $ be defined by \eqref{eq:ourGt},  $ \gamma >0 $ and 
		\begin{equation}\label{eq:Yt}
		\hat{Y} (t) \equiv  \Pi \frac{g^2}{\frac{\gamma}{2}  -i  ( \hat{H}_S - \varepsilon)} \left(1 - e^{- \left(\frac{\gamma}{2}  -i  ( \hat{H}_S - \varepsilon)\right)t} \right),
		\end{equation}
		then equation \eqref{eq:ourRedfield} takes the form
		\begin{equation}\label{eq:ourRedfieldOnePeak}
		\frac{d}{dt} \rho_{SI}(t)  = | 0 \rangle \langle 0 |   \mathrm{Tr}\; (\hat{Y}(t) \rho_{SI}(t) ) -   \hat{Y}(t)  \rho_{SI}(t)  + | 0 \rangle \langle 0 |  \mathrm{Tr}\; (  \rho_{SI}(t)  \hat{Y}^{\dagger}(t) )    - \rho_{SI}(t) \hat{Y}^{\dagger}(t)
		\end{equation}
		and equation \eqref{eq:ourGKSL} takes the form
		\begin{equation}\label{eq:ourGKSLOnePeak}
		\frac{d}{dt} \rho_{SI}(t)  = | 0 \rangle \langle 0 |   \mathrm{Tr}\; ((\hat{Y}(+\infty) +\hat{Y}^{\dagger}(+\infty)) )\rho_{SI}(t) ) -   \hat{Y}(+\infty)   \rho_{SI}(t)  - \rho_{SI}(t) \hat{Y}^{\dagger}(+\infty)),
		\end{equation}
		where 
		\begin{equation*}
		\hat{Y}(+\infty) = \Pi \frac{g^2}{\frac{\gamma}{2}  -i  (\hat{H}_S - \varepsilon)} = \gamma^{-1} \mathcal{J}(\hat{H}_S ) \left(\frac{\gamma}{2} I +i  (\hat{H}_S - \varepsilon I)\right).
		\end{equation*}
	\end{corollary}
	
	\begin{proof} 
		Actually, in general, \eqref{eq:ourRedfield} could be represented in form \eqref{eq:ourRedfieldOnePeak} with
		\begin{equation*}
		\hat{Y} (t) =  \Pi \int_0^t ds G(t-s) e^{i  \hat{H}_S(t-s)}.
		\end{equation*}
		In our case, when $ G(t) $ is defined by \eqref{eq:ourGt}, integration with respect to $ s $ leads to \eqref{eq:Yt}. In the case of \eqref{eq:ourGKSLOnePeak} one would integrate to infinity and obtain \eqref{eq:ourGKSLOnePeak}. 
	\end{proof}
	In particular, equation \eqref{eq:ourGKSL} has a time-independent generator, i.e. non-secular terms do not appear and it is already in the secular approximation. Actually, this is a result of the rotating wave approximation in initial interaction Hamiltonian \eqref{eq:H_I}. One could show \cite[p. 136]{Breuer02} that in general the Redfield equation in the secular approximation has the GKSL form, but let us present the GKSL form for special case \eqref{eq:ourGKSLOnePeak} explicitly:
	\begin{equation*}
	\frac{d}{dt} \rho_{SI}(t)  = -i[\gamma^{-1} \mathcal{J}(\hat{H}_S )  (\hat{H}_S - \varepsilon I), \rho_{SI}(t)] + \sum_{\alpha} \mathcal{J}(E_{\alpha}) \left(|0 \rangle \langle \alpha | \rho_{SI}(t)  | \alpha \rangle \langle 0 | - \frac12 \{| 0 \rangle \langle \alpha |,  \rho_{SI}(t) \} \right),
	\end{equation*}
	where the braces denote the anticommutator $ \{A, B\} \equiv A B + B A $.
	
	Let us note that similarly to $ \hat{H}_S $ the matrix $ \hat{Y}(t) $ defined by \eqref{eq:Yt} is supported on the subspace $ \mathbb{C}^N $ of $ \mathbb{C} \oplus \mathbb{C}^N $, i.e. $ \hat{Y}(t)  = 0 \oplus Y(t) $.
	\begin{corollary} 
		Equation \eqref{eq:ourRedfieldOnePeak} has a solution of form \eqref{eq:formOfSol}, where the vector-valued function $ | \psi_I(t) \rangle $ is the solution of the Cauchy problem
		\begin{equation}\label{eq:YpsiI}
		\frac{d}{dt} | \psi_I (t) \rangle = - Y(t) | \psi_I (t) \rangle,  \qquad | \psi_I (0) \rangle = | \psi (0) \rangle
		\end{equation}
		and the matrix-valued function $ \sigma(t) $ is a solution of  the Cauchy problem
		\begin{equation}\label{eq:Ysigma}
		\frac{d}{dt} \sigma(t) = - Y(t) \sigma(t) - \sigma(t) Y^{\dagger}(t), \qquad \sigma(0) = | \psi (0) \rangle \langle \psi (0) |.
		\end{equation}
		Analogously, equation  \eqref{eq:ourGKSLOnePeak} has a solution of form \eqref{eq:formOfSol}, where the vector-valued function $ | \psi_I(t) \rangle $ is a solution of the Cauchy problem
		\begin{equation}\label{eq:YpsiIGKSL}
		\frac{d}{dt} | \psi_I (t) \rangle = - Y(+\infty) | \psi_I (t) \rangle,  \qquad | \psi_I (0) \rangle = | \psi (0) \rangle, 
		\end{equation}
		and the matrix-valued function $ \sigma(t) $ is a solution of  the Cauchy problem
		\begin{equation}\label{eq:YsigmaGKSL}
		\frac{d}{dt} \sigma(t) = - Y(+\infty) \sigma(t) - \sigma(t) Y^{\dagger}(+\infty), \qquad \sigma(0) = | \psi (0) \rangle \langle \psi (0) |.
		\end{equation}
	\end{corollary}
	
	\begin{corollary}
		If one decomposes $ |\psi(t)\rangle  $ and  $\sigma(t)  $ in the global basis $ |\psi_I(t)\rangle = \sum_{\alpha}\psi_{\alpha}(t) | \alpha \rangle,  \sigma(t) = \sum_{\alpha\beta}\sigma_{\alpha \beta}(t) | \alpha \rangle \langle \beta | $, then Cauchy problems \eqref{eq:YpsiI}, \eqref{eq:Ysigma} take the form
		\begin{equation}\label{eq:Ypsialpha}
		\frac{d}{dt} \psi_{\alpha}(t) = - Y_{\alpha}(t)\psi_{\alpha}(t), \qquad \psi_{\alpha}(0)= \langle \alpha| \psi (0) \rangle, 
		\end{equation}
		\begin{equation}\label{eq:Ysigmaalpha}
		\frac{d}{dt} \sigma_{\alpha \beta}(t) = - (Y_{\alpha}(t) + Y_{\beta}^*(t)  ) \sigma_{\alpha \beta}(t), \qquad \sigma_{\alpha \beta}(0) = \psi_{\alpha}^*(0) \psi_{\beta}(0),
		\end{equation}
		accordingly, and Cauchy problems \eqref{eq:YpsiIGKSL}, \eqref{eq:YsigmaGKSL} take the form
		\begin{equation*}%\label{eq:YpsialphaGKSL}
		\frac{d}{dt} \psi_{\alpha}(t) = - Y_{\alpha}(+\infty)\psi_{\alpha}(t), \qquad \psi_{\alpha}(0)= \langle \alpha| \psi (0) \rangle, 
		\end{equation*}
		\begin{equation*}%\label{eq:YsigmaalphaGKSL}
		\frac{d}{dt} \sigma_{\alpha \beta}(t) = - (Y_{\alpha}(+\infty) + Y_{\beta}^*(+\infty) ) \sigma_{\alpha \beta}(t), \qquad \sigma_{\alpha \beta}(0) = \psi_{\alpha}^*(0) \psi_{\beta}(0),
		\end{equation*}
		accordingly, where
		\begin{equation*}
		Y_{\alpha}(t) \equiv Y_{\alpha}(+\infty) \left(1 - e^{- \left(\frac{\gamma}{2}  -i  \Delta E_{\alpha}\right)t} \right), Y_{\alpha}(+\infty) \equiv \gamma^{-1}
		\mathcal{J}(E_{\alpha}) \left(\frac{\gamma}{2}  +i  \Delta E_{\alpha} \right),
		\end{equation*}
		and the function $ \mathcal{J}(\omega)  $ is defined by formula \eqref{eq:spectralDensity}.
	\end{corollary}
	Similar to \eqref{eq:etaaEx} it is natural to define the decoherence rates between excited and ground states of the system by
	\begin{equation*}
	\eta_{\alpha 0}^{\rm R}(t) \equiv \mathrm{Re} \; Y_{\alpha}(t) = \gamma^{-1} \mathcal{J}(E_{\alpha}) \left( \frac{\gamma}{2} \left(1 - e^{- \frac{\gamma}{2} t} \cos (\Delta E_{\alpha} t) \right) + (E_{\alpha} - \varepsilon) e^{- \frac{\gamma}{2} t} \sin (\Delta E_{\alpha} t)\right)
	\end{equation*}
	and they are time-dependent. The population decay rates and the decoherence rates should be defined as follows
	\begin{equation}\label{eq:etaR}
	\eta_{\alpha}^{\rm R}(t) \equiv 2 \mathrm{Re} \; Y_{\alpha}(t), \qquad \eta_{\alpha \beta}^{\rm R}(t) \equiv \mathrm{Re} \; Y_{\alpha}(t) + \mathrm{Re} \; Y_{\beta}(t).
	\end{equation}
	At the large time $ t \rightarrow + \infty $ the rates predicted by the Redfield equation approach the rates predicted by the Redfield equation in the secular approximation
	\begin{equation}\label{eq:etaabGKSL}
	\begin{array}{ll}
	\eta_{\alpha 0}^{\rm GKSL}&\equiv \mathrm{Re} \; Y_{\alpha}(+\infty) =  \frac{1}{2} \mathcal{J}(E_{\alpha}) = \frac{\gamma}{2} \frac{g^2}{ \left(\frac{\gamma}{2}\right)^2 + \Delta E_{\alpha}^2}, \\
	\eta_{\alpha \beta}^{\rm GKSL} &\equiv \mathrm{Re} \;  Y_{\alpha}(+\infty) + \mathrm{Re} \; Y_{\beta}(+\infty),
	\end{array}
	\end{equation}
	\begin{equation}\label{eq:etaaGKSL}
	\eta_{\alpha}^{\rm GKSL} \equiv 2 \mathrm{Re} \; Y_{\alpha}(+\infty).
	\end{equation}
	
	Although, for the purposes of our article just the population decay rates and the decoherence rates are needed, one could obtain the explicit expression for the populations and coherences. Therefore let us state the following lemma.
	\begin{lemma}
		The solutions of Cauchy problems \eqref{eq:Ypsialpha} and \eqref{eq:Ysigmaalpha} have the form
		\begin{equation*}
		\psi_{\alpha}(t) =  \psi_{\alpha}(0) \exp \left( -\gamma^{-1} \mathcal{J}(E_{\alpha}) \left(\frac{\gamma}{2}  +i  \Delta E_{\alpha} \right) \left( t - \frac{1 - e^{- \left(\frac{\gamma}{2}  -i \Delta E_{\alpha}\right)t}}{\frac{\gamma}{2}  -i  \Delta E_{\alpha}}\right)\right),
		\end{equation*}
		\begin{equation*}
		\sigma_{\alpha \beta}(t) = \overline{\psi}_{\alpha}(t) \psi_{\beta}(t) ,
		\end{equation*}
		accordingly.
	\end{lemma}
	To prove that one should directly integrate  \eqref{eq:Ypsialpha} and \eqref{eq:Ysigmaalpha}. Let us note that this integration with respect to time could be regarded as an analog of time-deformation \cite{Filippov18}, but it depends on the global basis.
	
	\section{Comparison}\label{sec:Compare}
	First of all let us compare decoherence rates between excited and ground states for different equations mentioned in the previous sections.
	
	\begin{theorem}\label{th:CoherGroundEx}
		Let $ \eta_{\alpha 0}^{\rm exact}(\gamma, \Delta E_{\alpha}, g) $, $  \eta_{\alpha 0}^{\rm B}(\gamma, \Delta E_{\alpha}, g) $ and $ \eta_{\alpha 0}^{\rm GKSL}(\gamma, \Delta E_{\alpha}, g) $ be the real-valued functions defined for $ \gamma>0, g>0, \Delta  E_{\alpha} \in \mathbb{R} $ by  \eqref{eq:etaabEx},\eqref{eq:etaabB} and \eqref{eq:etaabGKSL}, accordingly, then
		\begin{equation}\label{eq:a0ExB}
		\eta_{\alpha 0}^{\rm exact}(\gamma, \Delta E_{\alpha}, g) = \eta_{\alpha 0}^{\rm B}(\gamma, \Delta E_{\alpha}, g),
		\end{equation}
		\begin{equation}\label{eq:a0ExGKSL}
		\begin{cases}
		\eta_{\alpha 0}^{\rm exact}(\gamma, \Delta E_{\alpha}, g) < \eta_{\alpha 0}^{\rm GKSL}(\gamma, \Delta E_{\alpha}, g), \qquad  F_1\left(\frac{\Delta E_{\alpha}^2}{g^2},\frac{\gamma^2}{g^2}\right) >0  \, \mathrm{or} \, \gamma <\sqrt{8} g,\\
		\eta_{\alpha 0}^{\rm exact}(\gamma, \Delta E_{\alpha}, g) = \eta_{\alpha 0}^{\rm GKSL}(\gamma, \Delta E_{\alpha}, g), \qquad  F_1\left(\frac{\Delta E_{\alpha}^2}{g^2},\frac{\gamma^2}{g^2}\right) = 0,  \, \mathrm{and} \, \gamma \geqslant \sqrt{8} g,\\
		\eta_{\alpha 0}^{\rm exact}(\gamma, \Delta E_{\alpha}, g) > \eta_{\alpha 0}^{\rm GKSL}(\gamma, \Delta E_{\alpha}, g), \qquad  F_1\left(\frac{\Delta E_{\alpha}^2}{g^2},\frac{\gamma^2}{g^2}\right) < 0  \, \mathrm{and} \, \gamma > \sqrt{8} g,
		\end{cases}
		\end{equation}
		where
		\begin{equation}\label{eq:F1}
		F_1(u,v) = -24576 u^3 v-10240 u^2 v^2+32768 u^2 v-512 u v^3+128 v^4-2048 v^3+8192 v^2.
		\end{equation}
	\end{theorem}
	The proof of \eqref{eq:a0ExB} follows immediately from lemma \ref{lem:BornApprNZsol} and the discussion after that. The proof of \eqref{eq:a0ExGKSL} is based on the direct comparison of \eqref{eq:etaabEx} and \eqref{eq:etaabGKSL} for the ground-excited case.
	
	\begin{corollary}
		Let $ \eta_{\alpha 0}^{\rm R}(\gamma, \Delta E_{\alpha}, g, t^*) $ be defined by \eqref{eq:etaR} for $ \gamma>0, g>0, \Delta  E_{\alpha} \in \mathbb{R} , t \geqslant 0$ and $ \eta_{\alpha 0}^{\rm exact}(\gamma, \Delta E_{\alpha}, g) $ be  defined by \eqref{eq:etaabEx} for $ \gamma>0, g>0, \Delta  E_{\alpha} \in \mathbb{R} $. Let  $ F_1(u,v) $ be defined by \eqref{eq:F1}. If $ F_1\left(\frac{\Delta E_{\alpha}^2}{g^2},\frac{\gamma^2}{g^2}\right) >0 $, then there exists $ t^* \in \mathbb{R}$ (may be not unique) such that  $ \eta_{\alpha 0}^{\rm R}(\gamma, \Delta E_{\alpha}, g, t^*) =  \eta_{\alpha 0}^{\rm exact}(\gamma, \Delta E_{\alpha}, g) $.
	\end{corollary}
	
	The proof follows from the fact that $ \eta_{\alpha 0}^{\rm R}(t) $ run over all the points from $ [0,\eta_{\alpha 0}^{\rm GKSL} ] $  (but may be do not only them) and theorem \ref{th:CoherGroundEx}.
	
	Now we are going to compare the population decay rates of excited states for  different equations mentioned in previous sections. For that we need a technical lemma based on Routh-Hurwitz criterion.
	
	\begin{lemma}\label{lem:qubicRGCondConst}
		Let $ x \in \mathbb{R} $, then the roots of the cubic polynomial
		\begin{equation}\label{eq:cubic}
		f(\lambda) = \lambda^3 + a_1 \lambda^2 + a_2 \lambda + a_3
		\end{equation}
		lie in the half-plane $ \mathrm{Re} \; \lambda < x $ if and only if
		\begin{equation}\label{eq:qubicRGCondConst}
		a_1 > -3 x, \qquad (a_1+3x)(a_2 +2 a_1x + 3x^2) > a_3 + a_2 x +a_1 x^2 + x^3 >0.
		\end{equation}
	\end{lemma}
	\begin{proof}
		Let us substitute $ \lambda = \tilde{\lambda} + x $, where $ x \in \mathbb{R}$, into \eqref{eq:cubic}
		\begin{equation*}
		f(\tilde{\lambda} + x) = \tilde{\lambda}^3 +(a_1+3x) \tilde{\lambda}^2 + (a_2 +2 a_1x + 3x^2) \tilde{\lambda} + (a_3 + a_2 x +a_1 x^2 + x^3).
		\end{equation*}
		Applying the Routh-Hurwitz criterion \cite[p. 194]{Gantmaher10} (in the particular case of a cubic equation it is also called the Vyshnegradsky criterion \cite[p. 37]{Postnikov81}) to this polynomial in $ \tilde{\lambda} $ we obtain: $ \mathrm{Re} \; \tilde{\lambda} < 0  $, i.e. $ \mathrm{Re} \; \lambda < x $,  if and only if \eqref{eq:qubicRGCondConst} are satisfied.
	\end{proof}

	\begin{theorem}\label{th:popDecay}
		Let $ \eta_{\alpha}^{\rm exact}(\gamma, \Delta E_{\alpha}, g) $, $  \eta_{\alpha}^{\rm B}(\gamma, \Delta E_{\alpha}, g) $ and $ \eta_{\alpha}^{\rm GKSL}(\gamma, \Delta E_{\alpha}, g) $ be the real-valued functions defined for $ \gamma>0, g>0, \Delta  E_{\alpha} \in \mathbb{R} $ by  \eqref{eq:etaaEx},\eqref{eq:etaaB} and \eqref{eq:etaaGKSL}, accordingly, then
		\begin{align*}
		& \begin{cases}
		\eta_{\alpha}^{\rm exact}(\gamma, \Delta E_{\alpha}, g) < \eta_{\alpha}^{\rm B} (\gamma, \Delta E_{\alpha}, g), & F_2\left(\frac{\Delta E_{\alpha}^2}{g^2},\frac{\gamma^2}{g^2}\right) >0, \\
		\eta_{\alpha}^{\rm exact}(\gamma, \Delta E_{\alpha}, g) = \eta_{\alpha}^{\rm B}(\gamma, \Delta E_{\alpha}, g), & F_2\left(\frac{\Delta E_{\alpha}^2}{g^2},\frac{\gamma^2}{g^2}\right) =0, \\
		\eta_{\alpha}^{\rm exact}(\gamma, \Delta E_{\alpha}, g) > \eta_{\alpha}^{\rm B}(\gamma, \Delta E_{\alpha}, g), & F_2\left(\frac{\Delta E_{\alpha}^2}{g^2},\frac{\gamma^2}{g^2}\right) <0,
		\end{cases}\\
		&\begin{cases}
		\eta_{\alpha}^{\rm exact}(\gamma, \Delta E_{\alpha}, g) < \eta_{\alpha}^{\rm GKSL}(\gamma, \Delta E_{\alpha}, g), \qquad  F_1\left(\frac{\Delta E_{\alpha}^2}{g^2},\frac{\gamma^2}{g^2}\right) >0  \, \mathrm{or} \, \gamma <\sqrt{8} g,\\
		\eta_{\alpha}^{\rm exact}(\gamma, \Delta E_{\alpha}, g) = \eta_{\alpha}^{\rm GKSL}(\gamma, \Delta E_{\alpha}, g), \qquad  F_1\left(\frac{\Delta E_{\alpha}^2}{g^2},\frac{\gamma^2}{g^2}\right) = 0,\\
		\eta_{\alpha}^{\rm exact}(\gamma, \Delta E_{\alpha}, g) > \eta_{\alpha}^{\rm GKSL}(\gamma, \Delta E_{\alpha}, g), \qquad  F_1\left(\frac{\Delta E_{\alpha}^2}{g^2},\frac{\gamma^2}{g^2}\right) < 0  \, \mathrm{and} \, \gamma > \sqrt{8} g,
		\end{cases}\\
		&\begin{cases}
		\eta_{\alpha}^{\rm B}(\gamma, \Delta E_{\alpha}, g) < \eta_{\alpha}^{\rm GKSL}(\gamma, \Delta E_{\alpha}, g), \qquad  &F_3\left(\frac{\Delta E_{\alpha}^2}{g^2},\frac{\gamma^2}{g^2}\right) >0 \, \mathrm{and} \,  F_4\left(\frac{\Delta E_{\alpha}^2}{g^2},\frac{\gamma^2}{g^2}\right) >0, \\
		\eta_{\alpha}^{\rm B}(\gamma, \Delta E_{\alpha}, g) = \eta_{\alpha}^{\rm GKSL}(\gamma, \Delta E_{\alpha}, g), \qquad  &F_3\left(\frac{\Delta E_{\alpha}^2}{g^2},\frac{\gamma^2}{g^2}\right) = 0, |\Delta E_{\alpha}| \leqslant\frac{\sqrt{3}}{2} g \, \mathrm{or}\\
		&\, F_4\left(\frac{\Delta E_{\alpha}^2}{g^2},\frac{\gamma^2}{g^2}\right) = 0, |\Delta E_{\alpha}| >\frac{\sqrt{3}}{2} g,\\
		\eta_{\alpha}^{\rm B}(\gamma, \Delta E_{\alpha}, g) > \eta_{\alpha}^{\rm GKSL}(\gamma, \Delta E_{\alpha}, g), \qquad  &F_3\left(\frac{\Delta E_{\alpha}^2}{g^2},\frac{\gamma^2}{g^2}\right) < 0 \, \mathrm{or} \,  F_4\left(\frac{\Delta E_{\alpha}^2}{g^2},\frac{\gamma^2}{g^2}\right) < 0,
		\end{cases}
		\end{align*}
		where
		\begin{align}\label{eq:F2}
		F_2(u,v) =& 2304 u^5+6400 u^4 v+30720 u^4+2912 u^3 v^2+40448 u^3 v+143104 u^3+400 u^2 v^3-\nonumber\\
		&-10112 u^2 v^2+35264 u^2 v+256000 u^2+9 u v^4-480 u
		v^3+9232 u v^2-67584 u v + \nonumber\\
		&+64512 u+36 v^3-1920 v^2+33984 v-200704,
		\end{align}
		\begin{align}\label{eq:F3}
		F_3(u,v) = 
		256 u^4&+256 u^3 v-256 u^3+96 u^2 v^2-704 u^2 v-1024 u^2+16 u v^3- \nonumber\\
		&-304 u v^2+1536 u v+v^4 -36 v^3+448 v^2-2048 v,
		\end{align}
		\begin{equation}\label{eq:F4}
		F_4(u,v) = -16 u^2 - 8 v + v^2,
		\end{equation}
		and $ F_1(u,v) $ is defined by formula \eqref{eq:F1}. 
	\end{theorem}
	
	\begin{proof}
		1) To compare $ \eta_{\alpha}^{\rm exact} $ and $ \eta_{\alpha}^{\rm B} $ let us apply lemma \ref{lem:qubicRGCondConst} to the equation $ f_{\alpha}(\lambda) = 0 $, where $ f_{\alpha}(\lambda)  $ is defined by formula  \eqref{eq:falpha}, i.e.
		\begin{equation*}
		a_1 = \gamma, \quad a_2 = \left(\frac{\gamma}{2}\right)^2 +2 g^2 + \Delta E_{\alpha}^2, \quad a_3 =  g^2 \gamma,
		\end{equation*}
		and let $ x=\mathrm{Re}\; \lambda_{\alpha, \pm}^{\rm exact} $ which is defined by \eqref{eq:Relambdsalpha}. The third of inequalities \eqref{eq:qubicRGCondConst} appears to be satisfied without further assumptions and the first one follows from the second one. Thus, only the second inequality has to be satisfied. It takes the form
		\begin{equation*}
		F_2\left(\frac{ \Delta E_{\alpha}^2}{g^2},\frac{\gamma^2}{g^2}\right) >0,
		\end{equation*}
		where the function $ F_2(u,v) $ is defined by formula \eqref{eq:F2}.
		
		2) The comparison of $ \eta_{\alpha}^{\rm exact} = 2 \eta_{\alpha 0}^{\rm exact} $ and $ \eta_{\alpha}^{\rm GKSL} =  2 \eta_{\alpha 0}^{\rm GKSL}$ is equivalent to the comparison of $  \eta_{\alpha 0}^{\rm exact} $ and $  \eta_{\alpha 0}^{\rm GKSL} $, which was done in theorem \ref{th:CoherGroundEx}.
		
		3) The comparison of $ \eta_{\alpha}^{\rm GKSL} $ and $ \eta_{\alpha}^{\rm B} $  is similar to the first one. The only distinction is that one should assume $ x =-\eta_{\alpha 0}^{\rm GKSL} $, then functions \eqref{eq:F3} and \eqref{eq:F4} occur from lemma \ref{lem:qubicRGCondConst}.
	\end{proof}
	
	\begin{figure}
		\centering
		\includegraphics[width=.4\linewidth]{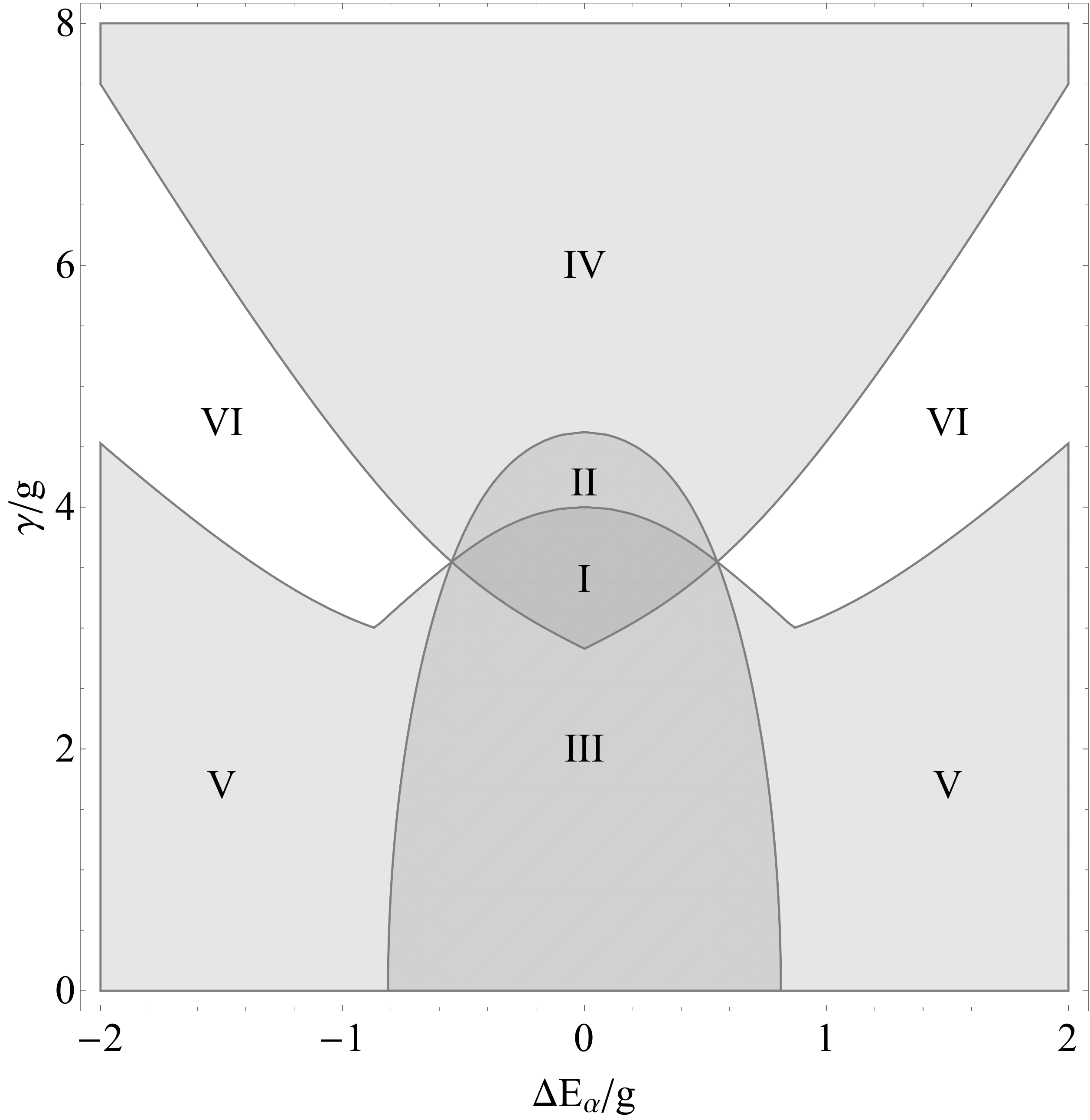}
		\caption{The following inequalities for the population decay rates predicted by different equations are held in the regions denoted by Roman numerals (boundaries are not included). I:~$ \eta_{\alpha}^{\rm exact} > \eta_{\alpha}^{\rm GKSL} > \eta_{\alpha}^{\rm B} $, II:~$ \eta_{\alpha}^{\rm exact} > \eta_{\alpha}^{\rm B} > \eta_{\alpha}^{\rm GKSL} $, III:~$ \eta_{\alpha}^{\rm GKSL} >  \eta_{\alpha}^{\rm exact} > \eta_{\alpha}^{\rm B} $, IV:~$ \eta_{\alpha}^{\rm B} > \eta_{\alpha}^{\rm exact} > \eta_{\alpha}^{\rm GKSL}  $, V:~$  \eta_{\alpha}^{\rm GKSL} > \eta_{\alpha}^{\rm B} > \eta_{\alpha}^{\rm exact} $, VI:~$ \eta_{\alpha}^{\rm B} > \eta_{\alpha}^{\rm GKSL} > \eta_{\alpha}^{\rm exact} $. The inequalities for the excited-ground coherences are: I, II, IV: $ \eta_{\alpha 0}^{\rm exact} = \eta_{\alpha 0}^{\rm B} < \eta_{\alpha 0}^{\rm GKSL} $, III, V, VI: $  \eta_{\alpha 0}^{\rm exact} = \eta_{\alpha 0}^{\rm B} > \eta_{\alpha 0}^{\rm GKSL} $.}
		\label{fig:MPic}
	\end{figure}
	
	The results of theorems \ref{th:CoherGroundEx} and \ref{th:popDecay} are presented in figure \ref{fig:MPic}.  The following feature should be noted.
	
	\begin{corollary}
		There are only two points of the half-plane $ \gamma/g >0, \Delta E_{\alpha}/g \in \mathbb{R} $ such that
		\begin{equation*}
		\eta_{\alpha}^{\rm exact}(\gamma, \Delta E_{\alpha}, g) = \eta_{\alpha}^{\rm B}(\gamma, \Delta E_{\alpha}, g) = \eta_{\alpha}^{\rm GKSL}(\gamma, \Delta E_{\alpha}, g).
		\end{equation*}
		These points are approximately $ \Delta E_{\alpha}/g \approx \pm 0.55 $, $ \gamma/g \approx 3.55 $.
	\end{corollary}
	If one wants to have simpler conditions for the comparison of $ \eta_{\alpha}^{\rm exact}$ and  $\eta_{\alpha}^{\rm B} $ than defined by \eqref{eq:F2}, then one could use the following proposition.
	\begin{corollary}
		If $ \gamma> \sqrt{64/3} g \approx 4.62 g$ or $ |\Delta E_{\alpha}|> \sqrt{\frac{1}{3} \left(-4+\sqrt[3]{44-3 \sqrt{177}}+\sqrt[3]{44+3 \sqrt{177}}\right)} g \approx 0.81 g $, then $ \eta_{\alpha}^{\rm exact}(\gamma, \Delta E_{\alpha}, g) < \eta_{\alpha}^{\rm B} (\gamma, \Delta E_{\alpha}, g) $. If $ \gamma \leqslant 3 \sqrt{2 g^2 - 4 \Delta E_{\alpha}^2} $, then $ \eta_{\alpha}^{\rm exact}(\gamma, \Delta E_{\alpha}, g) > \eta_{\alpha}^{\rm B} (\gamma, \Delta E_{\alpha}, g) $.
	\end{corollary}
	
	\begin{figure}
		\centering
		\includegraphics[width=.4\linewidth]{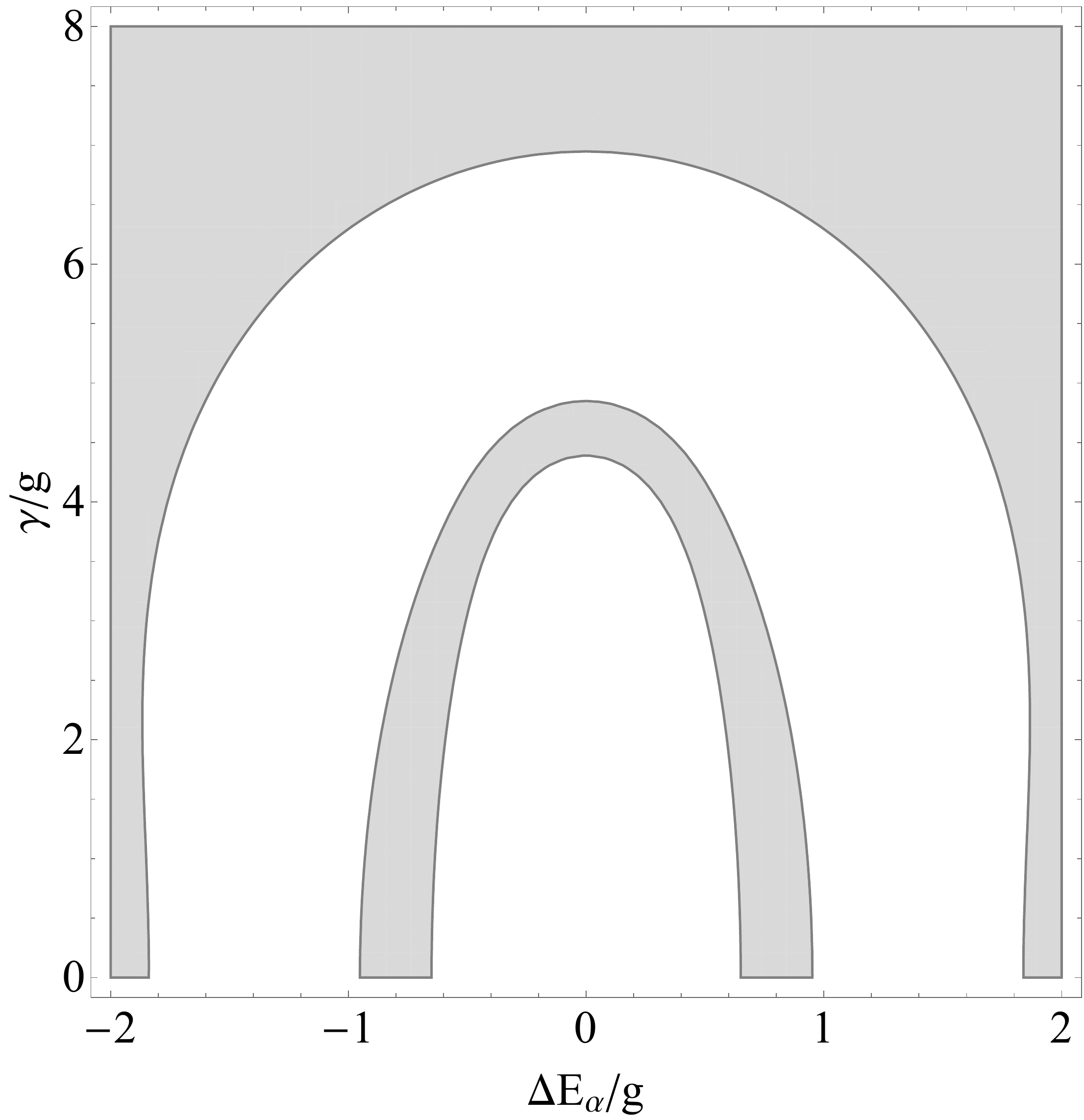} \quad
		\includegraphics[width=.4\linewidth]{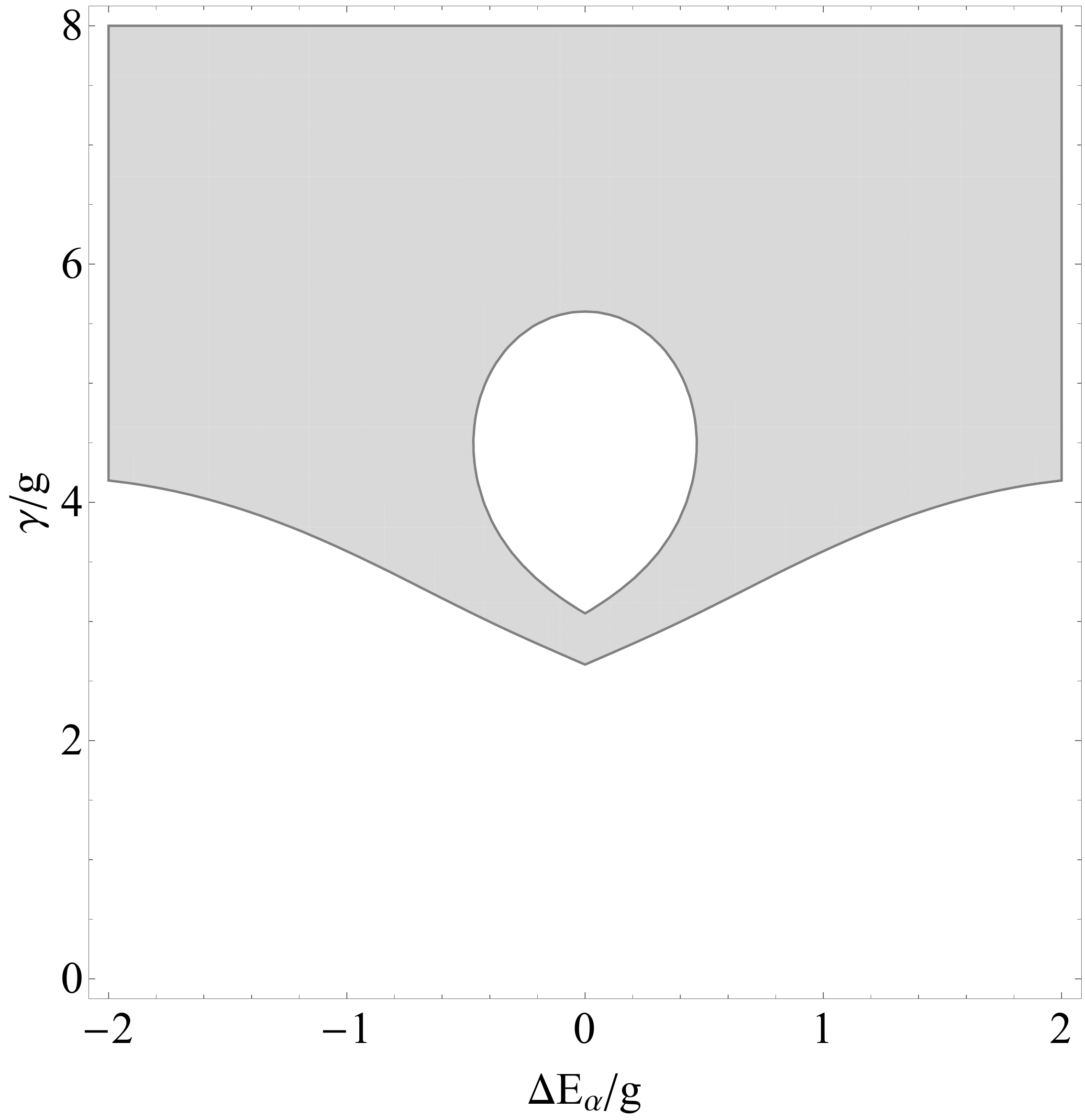}
		\caption{The regions, where the following inequalities are held, are filled: $ |\eta_{\alpha}^{\rm exact}-\eta_{\alpha}^{\rm B}| < 0.15 \eta_{\alpha}^{\rm exact} $ (left),  $ |\eta_{\alpha}^{\rm exact}-\eta_{\alpha}^{\rm GKSL}| < 0.15 \eta_{\alpha}^{\rm exact} $ (right).}
		\label{fig:near}
	\end{figure}
	
	Figure \ref{fig:MPic} shows, where one population decay rates are greater than other ones, but it does not show how close they are. So we have done it numerically in figure \ref{fig:near}. One could see that $ \eta_{\alpha}^{\rm B} $ is close to  $ \eta_{\alpha}^{\rm exact} $ not only near the curve $ \eta_{\alpha}^{\rm B} = \eta_{\alpha}^{\rm exact} $ but also for scientifically small $ g $. The region, where $ \eta_{\alpha}^{\rm B} $ is close to  $ \eta_{\alpha}^{\rm GKSL} $, is also wide and not concentrated only near  the curve $ \eta_{\alpha}^{\rm B} = \eta_{\alpha}^{\rm GKSL} $. Interestingly the second region is not even close to be a subset of the first one, so there are areas of parameters when GKSL equations could provide the better fit than the non-Markovian Nakajima-Zwanzig equation in the Born approximation for the populations decay rate which is experimentally observable. But for a sufficiently narrow peak in the spectral density (small $ \gamma $) the non-Markovian Nakajima-Zwanzig equation in the Born approximation reproduces the population decay better than the GKSL one.
	
	\section{Conclusions}
	
	We have considered the model of the multi-level system interacting with several local reservoirs at zero temperature. We have compared the population decay rates and the decoherence rates for exact solution and several approximate master equations: the Nakajima-Zwanzig equation in the Born approximation, the Redfield equation. It was shown:
	\begin{enumerate}
		\item Both the initial model and the approximate master equations are exactly solvable in the global basis.
		\item The Nakajima-Zwanzig equation in the Born approximation gives an exact result for the coherences between the excited states and the ground states (without additional assumptions for the spectral density of the reservoir).
		\item The conditions for all possible inequalities between excited-ground decoherence rates and population rates in the global basis for the exact, Born and Markovian Redfield cases are fully characterized by theorems \ref{th:CoherGroundEx} and \ref{th:popDecay}.
		\item Both numerically and analytically we have shown that there exist the cases when the Markovian GKSL equation reproduces the population decay better than the non-Markovian Nakajima-Zwanzig equation in the Born approximation, but this is not the case for the sufficiently narrow spectral density.
	\end{enumerate}
	
	In our opinion the following directions for the further studies could be fruitful.
	\begin{enumerate}
		\item Application to the real systems. As in \cite{Teret19} in this study we were inspired by vibronic non-Markovian phenomena in light harvesting complexes. The approach described here could be applied to the one-exciton models \cite{Bradler10, Rebentrost09} of the Fenna-Matthews-Olson complexes at cryogenic temperatures. For them the non-Markovian phenomena were experimentally observed \cite{Engel07}, which leads to the sufficient interest in the quantum phenomena in photosynthetic systems \cite{Lee07, Mohseni08, Plenio08, Collini10, Scholes11}.
		\item Finite-temperature analysis. The fact that \eqref{eq:IntOfM} is an integral of motion for our system allows one to separate the equations with fixed number of particles. So may be the exact finite-temperature solutions could be obtained on this way.
		\item Multiple Lorentzian and non-Lorentzian generalization of the results described. Multiple Lorentzian peaks case for the spectral density could be considered in a straight forward way by the methods from \cite{Imamoglu94, Garraway96, Garraway97, Garraway97a, Dalton01, Garraway06, Teret19}. Non-Lorentzian case could be dealt with by general Laplace transform methods, but we think that the approach from \cite{Wilkie00, Chen82} could provide more physical insight.
	\end{enumerate}
	
	\section{Acknowledgments}
	The author thanks A.\,S.~Trushechkin for sufficient help in setting the main goals of this study and fruitful discussion at all the steps of the study.  The author thanks I.\,V.~Volovich, S.\,V.~Kozyrev, B.\,O.~Volkov  and A.\,I.~Mikhailov for the useful discussion of the problems considered in the work.


\begin{thebibliography}{10}
	
	\bibitem{Nakajima58}
	
	S.~Nakajima, \textquotedblleft On Quantum Theory of Transport Phenomena: Steady Diffusion,\textquotedblright\;Progress of Theor. Phys. \textbf{20} (6), 948--959 (1958).
	
	\bibitem{Zwanzig60}
	
	R.~Zwanzig, \textquotedblleft Ensemble Method in the Theory of Irreversibility,\textquotedblright\;J. of Chem. Phys. \textbf{33} (5), 1338--1341 (1960).	
	
	\bibitem{Breuer02}
	
	H.-P.~Breuer, F.~Petruccione, \emph{The theory of open quantum systems} (Oxford University Press, Oxford, 2002).
	
	\bibitem{Carmichael13}
	 
	H.J.~Carmichael, \emph{Statistical methods in quantum optics 1: Master equations and Fokker-Planck equations} (Springer-Verlag Berlin, 2013). 
	
	\bibitem{Chruscinski10}
	
	D.~Chruscinski, A.~Kossakowski, \textquotedblleft General form of quantum evolution,\textquotedblright\;arXiv:1006.2764 (2010).
	
	\bibitem{ChruscinskiKossakowski10}
	
	D.~Chruscinski, A.~Kossakowski, \textquotedblleft Non-Markovian quantum dynamics: local versus nonlocal,\textquotedblright\;Phys. Rev. Lett. \textbf{104} (7), 070406 (2010).
	
	\bibitem{Valkunas13}
	
	L.~Valkunas, D.~Abramavicius, T.~Mancal,  \emph{Molecular excitation dynamics and relaxation: quantum theory and spectroscopy} (Wiley-VCH Verlag GmbH \& Co. KGaA, Weinheim, 2013).
	
	\bibitem{Amerongen18}
	
	H.~Van~Amerongen et al., \emph{Light harvesting in photosynthesis} (CRC Press, Boca Raton, 2018).
	
	\bibitem{Singh12}
	
	N.~Singh, P.~Brumer, \textquotedblleft Efficient computational approach to the non-Markovian second order quantum master equation: electronic energy transfer in model photosynthetic systems,\textquotedblright\; Mol. Phys., \textbf{110} (15-16), 1815--1828 (2012).
	
	\bibitem{Redfield65}
	 
	A.G.~Redfield, \textquotedblleft The Theory of Relaxation Processes,\textquotedblright\;Adv. in Magnetic and Optical Resonance \textbf{1}, 1--32 (1965).
	
	\bibitem{Ishizaki08}
	
	A.~Ishizaki, Y.~Tanimura, \textquotedblleft Nonperturbative non-Markovian quantum master equation: Validity and limitation to calculate nonlinear response functions,\textquotedblright\; Chem. Phys., \textbf{347} (1-3), 185--193 (2008).
	
	\bibitem{May08}
	
	V.~May, O.~Kuhn, \emph{Charge and energy transfer dynamics in molecular systems} (Wiley-VCH Verlag GmbH \& Co. KGaA, Weinheim, 2008).
	
	\bibitem{Lindblad76}
	 
	G.~Lindblad, \textquotedblleft On the generators of quantum dynamical semigroups,\textquotedblright\;Comm. in Math. Phys., \textbf{48} (2), 119--130 (1976).
	
	\bibitem{Gorini76}
	 
	V.~Gorini, A.~Kossakowski, E.C.G.~Sudarshan, \textquotedblleft Completely positive dynamical semigroups of N-level systems\textquotedblright\;J. of Math. Phys., \textbf{17} (5), 821--825 (1976).
	
	\bibitem{Davies1974}
	 
	E.B.~Davies, \textquotedblleft Markovian master equations,\textquotedblright\;Comm. in Math. Phys. \textbf{39} (2), 91--110 (1974).
	
	\bibitem{Accardi2002}
	
	L.~Accardi, Y.G.~Lu, I.~Volovich, \emph{Quantum theory and its stochastic limit} (Springer, Berlin, 2002).
	
	\bibitem{Krylov70}
	 
	N.M.~Krylov, N.N.~Bogoliubov, \textquotedblleft Fokker–Planck equations generated in perturbation theory by a method based on the spectral properties of a perturbed Hamiltonian,\textquotedblright\;Zapiski Kafedry Fiziki Akademii Nauk Ukrainian SSR \textbf{4} 81--157 (1939). [in Ukrainian].
	
	\bibitem{VanHove55}
	
	L.~Van~Hove, \textquotedblleft Quantum-mechanical perturbations giving rise to a statistical transport equation,\textquotedblright\; Physica \textbf{21}, 517 (1955).
	
	\bibitem{Suarez92}
	 
	A.~Suarez, R.~Silbey, I.~Oppenheim, \textquotedblleft Memory effects in the relaxation of quantum open systems,\textquotedblright\;J. of Chem. Phys. \textbf{97} (7), 5101--5107 (1992).
	
	\bibitem{Pierre99}
	 
	P.~Gaspard, N.~Masataka, \textquotedblleft Slippage of initial conditions for the Redfield master equation,\textquotedblright\;J. of Chem. Phys. \textbf{111} (13) 5668--5675 (1999).
	
	\bibitem{Anderloni07}
	 
	S.~Anderloni, F.~Benatti, R.~Floreanini, \textquotedblleft Redfield reduced dynamics and entanglement,\textquotedblright\;J. of Phys. A \textbf{40} (7), 1625--1632 (2007).
	
	\bibitem{Giovannetti19}
	
	D.~Farina, V.~Giovannetti, \textquotedblleft Open Quantum System Dynamics: recovering positivity of the Redfield equation via Partial-Secular Approximation,\textquotedblright\; arXiv:1903.07324 (2019).
	
	\bibitem{Novoderezhkin04}
	
	V.I.~Novoderezhkin et al., \textquotedblleft Coherent Nuclear and Electronic Dynamics in Primary Charge Separation in Photosynthetic Reaction Centers: A Redfield Theory Approach,\textquotedblright\; J. of Phys. Chem. B, \textbf{108} (22), 7445--7457 (2004).
	
	\bibitem{Novoderezhkin03}
	
	V.I.~Novoderezhkin, M.~Wendling, R.~Van~Grondelle,  Intra-and Interband Transfers in the B800--B850 Antenna of Rhodospirillum molischianum: Redfield Theory Modeling of Polarized Pump--Probe Kinetics. J. of Phys. Chem. B, \textbf{107} (41), 11534--11548 (2003).
	
	\bibitem{Purkayastha16}
	
	A.~Purkayastha, Ab.~Dhar, M.~Kulkarni, \textquotedblleft Out-of-equilibrium open quantum systems: A comparison of approximate quantum master equation approaches with exact results,\textquotedblright\;Phys. Rev. A \textbf{93} (6), 062114 (2016).
	
	\bibitem{Dodin18}
	
	A.~Dodin, T.~Tscherbul, R.~Alicki, A.~Vutha, P.~Brumer, \textquotedblleft Secular versus nonsecular Redfield dynamics and Fano coherences in incoherent excitation: An experimental proposal,\textquotedblright\;Phys. Rev. A \textbf{97} (1), 013421 (2018). 
	
	\bibitem{Kohen97}
	
	D.~Kohen, C.C.~Marston, D.J.~Tannor, \textquotedblleft Phase space approach to theories of quantum dissipation,\textquotedblright\; The J. of Chem. Phys. \textbf{107} (13), 5236--5253 (1997).
	
	\bibitem{Lindblad76Br}
	
	G.~Lindblad, \textquotedblleft Brownian motion of a quantum harmonic oscillator,\textquotedblright\; Rep. on Math. Phys., \textbf{10} (3), 393--406 (1976).
	
	\bibitem{Imamoglu94}
	 
	A.~Imamog, \textquotedblleft Stochastic wave-function approach to non-Markovian systems,\textquotedblright\;Phys. Rev. A \textbf{50} (5), 3650 (1994).
	
	\bibitem{Garraway96}
	 
	B.M.~Garraway, P.L.~Knight, \textquotedblleft Cavity modified quantum beats,\textquotedblright\;Phys. Rev. A, \textbf{54} (4), 3592 (1996).
	
	\bibitem{Garraway97}
	 	
	B.M.~Garraway, \textquotedblleft Nonperturbative decay of an atomic system in a cavity,\textquotedblright\;Phys. Rev. A \textbf{55} (3), 2290 (1997).
	
	\bibitem{Garraway97a}
	 		
	B.M.~Garraway, \textquotedblleft Decay of an atom coupled strongly to a reservoir,\textquotedblright\;Phys. Rev. A \textbf{55} (6), 4636 (1997).
	
	\bibitem{Dalton01}
	 		
	B.J.~Dalton, S.M.~Barnett, B.M.~Garraway, \textquotedblleft Theory of pseudomodes in quantum optical processes,\textquotedblright\;Phys. Rev. A \textbf{64} (5), 053813 (2001). 
	
	\bibitem{Garraway06}
	 			
	B.M.~Garraway, B.J.~Dalton, \textquotedblleft Theory of non-Markovian decay of a cascade atom in high-Q cavities and photonic band gap materials,\textquotedblright\;J. of Phys. B \textbf{39} (15), S767 (2006).
	
	\bibitem{Teret19}
	 
	A.E.~Teretenkov, \textquotedblleft Pseudomode approach and vibronic non-Markovian phenomena in light harvesting complexes,\textquotedblright\;Proc. Steklov Inst. Math. \textbf{306} (2019), to be published; arXiv:1904.01430.
	
	\bibitem{Friedrichs48}
	
	K.O.~Friedrichs, \textquotedblleft On the perturbation of continuous spectra,\textquotedblright\;Comm. on Pure and Applied Math. \textbf{1} (4), 361--406 (1948).
	
	\bibitem{Kossakowski07}
	
	A.~Kossakowski, R.~Rebolledo, \textquotedblleft On non-Markovian time evolution in open quantum systems,\textquotedblright\;Open Sys. and Inform. Dyn. \textbf{14} (3), 265--274 (2007).
	
	\bibitem{Levy14}
	 
	A.~Levy, R.~Kosloff, \textquotedblleft The local approach to quantum transport may violate the second law of thermodynamics,\textquotedblright\;EPL, \textbf{107} (2), 20004 (2014).
	
	\bibitem{Trushechkin16}
	 
	A.S.~Trushechkin, I.V.~Volovich, \textquotedblleft Perturbative treatment of inter-site couplings in the local description of open quantum networks,\textquotedblright\;EPL, \textbf{113} (3), 30005 (2016).
	
	\bibitem{Kozyrev17}
	 
	S.V.~Kozyrev  et al., \textquotedblleft Flows in non-equilibrium quantum systems and quantum photosynthesis,\textquotedblright\;Inf. Dim. Anal., Quant. Prob. and Rel. Top. \textbf{20} (4), 1750021 (2017).
	
	\bibitem{Fleming10}
	 
	C.~Fleming et al., \textquotedblleft The rotating-wave approximation: consistency and applicability from an open quantum system analysis,\textquotedblright\;J. of Phys. \textbf{43} (40), 405304 (2010).
	
	\bibitem{Tang13}
	 
	N.~Tang, T.-T.~Xu, H.-S.~Zeng, \textquotedblleft Comparison between non-Markovian dynamics with and without rotating wave approximation,\textquotedblright\;Chinese Phys. B \textbf{22} (3), 030304 (2013).
	
	\bibitem{Breuer07}
	
	H.P.~Breuer, \textquotedblleft Non-Markovian generalization of the Lindblad theory of open quantum systems,\textquotedblright\;Phys. Rev. A, \textbf{75} (2), 022103 (2007).
	
	\bibitem{Chernega13}
	
	V.N.~Chernega, O.V.~Man'ko, V.I.~Man'ko, \textquotedblleft Generalized qubit portrait of the qutrit-state density matrix,\textquotedblright\;J. of Russ. Laser Research, \textbf{34} (4), 383--387  (2013).
	
	\bibitem{Chernega14}
	
	V.N.~Chernega, O.V.~Man'ko, V.I.~Man'ko, \textquotedblleft New inequality for density matrices of single qudit states,\textquotedblright\;J. of Russ. Laser Research, \textbf{35} (5), 457--461 (2014).
	
	\bibitem{Chernega14a}
	
	V.N.~Chernega, O.V.~Man'ko, V.I.~Man'ko, \textquotedblleft Subadditivity Condition for Spin Tomograms and Density Matrices of Arbitrary Composite and Noncomposite Qudit Systems,\textquotedblright\;J. of Russ. Laser Research,\textquotedblright\; \textbf{35} (3), 278-290. (2014).
	
	\bibitem{Manko18}
	
	V.I.~Manko, T.~Sabyrgaliyev, \textquotedblleft New entropic inequalities for qudit (spin j= 9/2),\textquotedblright\; arXiv:1807.00389 (2018).
	
	\bibitem{Antoniou03}
	 
	I.~Antoniou, E.~Karpov, G.~Pronko, E.~Yarevsky, \textquotedblleft Oscillating decay of an unstable system,\textquotedblright\;International Journal of Theoretical Physics, \textbf{42} (10), 2403-2421 (2003).
	
	\bibitem{Feshbach58}
	
	H.~Feshbach, \textquotedblleft Unified theory of nuclear reactions,\textquotedblright\;Ann. of Phys., \textbf{5} (4), 357--390 (1958).
	
	\bibitem{Feshbach62}
	
	H.~Feshbach, \textquotedblleft A unified theory of nuclear reactions. II,\textquotedblright\;Ann. of Phys., \textbf{19} (2), 287--313 (1962).
	
	\bibitem{Chruscinski13}
	
	D.~Chruscinski, A.~Kossakowski, \textquotedblleft Feshbach projection formalism for open quantum systems,\textquotedblright\;Phys. Rev. Lett. \textbf{111} (5), 050402 (2013).
	
	\bibitem{Fermi54}
	
	E.~Fermi, \textquotedblleft Polarization of high energy protons scattered by nuclei,\textquotedblright\; Nuovo Cimento \textbf{11}, 407	(1954).
	
	\bibitem{Feshbach54}
	
	H.~Feshbach, C.E.~Porter, and V.F.~Weisskopf, \textquotedblleft Model for nuclear reactions with neutrons,\textquotedblright\; Phys. Rev. \textbf{96}, 448 (1954).
	
	\bibitem{Luchnikov19}
	 
	I.A.~Luchnikov  et al., \textquotedblleft Machine learning of Markovian embedding for non-Markovian quantum dynamics,\textquotedblright\;arXiv:1902.07019 (2019).
	
	\bibitem{Sachdev84}
	 
	S.~Sachdev, \textquotedblleft Atom in a damped cavity,\textquotedblright\;Phys. Rev. A \textbf{29} (5), 2627 (1984).
	
	\bibitem{Weiss99}
	 
	U.~Weiss, \emph{Quantum dissipative systems}  (World Scientific, Singapore, 1999).
	
	\bibitem{Trushechkin19}
	 
	A.S.~Trushechkin, \textquotedblleft Calculation of coherences in Foerster and modified Redfield theories of excitation energy transfer,\textquotedblright\;arXiv:1902.00554 (2019).
	
	\bibitem{Burton05}
	
	T.A.~Burton, \emph{Volterra integral and differential equations} (Elsevier, Amsterdam, 2005).
	
	\bibitem{Meier99}
	
	C.~Meier, D.J.~Tannor, \textquotedblleft Non-Markovian evolution of the density operator in the presence of strong laser fields,\textquotedblright\; J. of Chem. Phys. \textbf{111} (8), 3365--3376 (1999). 
	
	\bibitem{Pomyalov10}
	
	A.~Pomyalov, C.~Meier, D.J.~Tannor, \textquotedblleft The importance of initial correlations in rate dynamics: A consistent non-Markovian master equation approach,\textquotedblright\; Chem. Phys. \textbf{370} (1-3), 98--108 (2010).
	
	\bibitem{Filippov18}
	
	S.N.~Filippov, D.~Chruscinski, \textquotedblleft Time deformations of master equations,\textquotedblright\; Phys. Rev. A \textbf{98} (2), 022123 (2018). 
	
	\bibitem{Gantmaher10}
	 
	F.R.~Gantmacher \emph{The Theory of Matrices, Vol. 2} (American Mathematical Society, Providence, 2000)
	
	\bibitem{Postnikov81}
	
	M.M.~Postnikov,  \emph{Stable polynomials} (Moscow, Nauka, 1981) [in Russian].
	
	\bibitem{Rebentrost09}
	
	P.~Rebentrost et al., \textquotedblleft Environment-assisted quantum transport,\textquotedblright\; New J. of Phys. \textbf{11} (3), 033003 (2009).
	
	\bibitem{Bradler10}
	
	K.~Bradler  et al., \textquotedblleft Identifying the quantum correlations in light-harvesting complexes,\textquotedblright\; Phys. Rev. A \textbf{82} (6), 062310 (2010).
	
	\bibitem{Engel07}
	
	G.S.~Engel  et al., \textquotedblleft Evidence for wavelike energy transfer through quantum coherence in photosynthetic systems,\textquotedblright\; Nature \textbf{446} (7137), 782 (2007).
	
	\bibitem{Lee07}
	
	H.~Lee, Y.C.~Cheng, G.R.~Fleming, \textquotedblleft Coherence dynamics in photosynthesis: protein protection of excitonic coherence,\textquotedblright\; Science \textbf{316} (5830), 1462--1465 (2007).
	
	\bibitem{Mohseni08}
	
	M.~Mohseni et al., \textquotedblleft Environment-assisted quantum walks in photosynthetic energy transfer,\textquotedblright\;  The J. of Chem. Phys. \textbf{129} (17), 11B603 (2008).
	
	\bibitem{Plenio08}
	
	M.B.~Plenio, S.F.~Huelga, \textquotedblleft Dephasing-assisted transport: quantum networks and biomolecules,\textquotedblright\;  New J. of Phys. \textbf{10}(11), 113019 (2008).
	
	\bibitem{Collini10}
	
	E.~Collini  et al., \textquotedblleft Coherently wired light-harvesting in photosynthetic marine algae at ambient temperature,\textquotedblright\; Nature \textbf{463} (7281), 644 (2010).
	
	\bibitem{Scholes11}
	
	Scholes G.D. et al., \textquotedblleft Lessons from nature about solar light harvesting,\textquotedblright\; Nature Chem. \textbf{3} (10). 763 (2011).
	
	\bibitem{Wilkie00}
	
	J.~Wilkie, \textquotedblleft Positivity preserving non-Markovian master equations,\textquotedblright\; Phys. Rev. E \textbf{62} (6), 8808 (2000).
	
	\bibitem{Chen82}
	
	G.~Chen, R.~Grimmer, \textquotedblleft Integral equations as evolution equations,\textquotedblright\; J. of Diff. Eq. \textbf{45} (1), 53--74 (1982).
	
\end{thebibliography}
\end{document}